\documentclass[11pt]{article}
\usepackage[utf8]{inputenc}
\usepackage{physics}

\usepackage{fullpage}
\usepackage{amsmath,amsthm,amsfonts,mathtools}
\usepackage{changepage}
\usepackage[nottoc,notlot,notlof]{tocbibind}
\usepackage{tikz}
\usetikzlibrary{quantikz} 
\usepackage{bbm} 
\usepackage{multicol} 
\usepackage{setspace,adjustbox,array}
\usepackage{tabularx,xifthen,afterpage,booktabs}
\newcolumntype{L}{>{\centering\arraybackslash}m{4cm}} 
\usepackage{tabularx}
\usepackage{soul}

\usepackage{paralist}
\usepackage[shortlabels]{enumitem}
  \setlist[itemize]{leftmargin=*}
  \setlist[enumerate]{leftmargin=*}
\usepackage{hyperref}
\usepackage[capitalize,nameinlink]{cleveref}
\usepackage[linesnumbered,boxed,noend,noline]{algorithm2e}
\usepackage{framed}
\usepackage{algpseudocode}
\usetikzlibrary{positioning,matrix,calc,arrows,shapes,fit,decorations,angles,quotes}
\usepackage{thmtools}
\usepackage{thm-restate}

\usepackage{tcolorbox}
\usepackage[caption=false]{subfig}

\usepackage[backend=biber,
            style=trad-alpha,
            backref=true,
            backrefstyle=three,
            doi=false,
            isbn=false,
            url=false,
            maxbibnames=20,
            maxcitenames=2]{biblatex}
\bibliography{refs,reductionsrefs}
\DefineBibliographyStrings{english}{%
  backrefpage = {p.\ },
  backrefpages = {pp.\ }
}

\DeclareMathOperator*{\E}{\mathbb{E}}
\newcommand{\eps}{\varepsilon}
\renewcommand{\epsilon}{\varepsilon}
\newcommand{\N}{\mathbb{N}}
\newcommand{\F}{\mathbb{F}}
\newcommand{\R}{\mathbb{R}}
\newcommand{\C}{\mathbb{C}}

\newcommand{\seq}{\subseteq}

\newcommand{\poly}{\mathrm{poly}}

\renewcommand{\sp}[1]{{\rm span}\{#1\}}

\renewcommand{\hat}{\widehat}
\renewcommand{\sp}{\mathsf{span}}
\newcommand{\ALG}{\mathsf{ALG}}

\newcommand{\Spec}{{\mathrm{Spec}}}

\renewcommand{\ip}[1]{{\langle #1 \rangle}}
\newcommand{\qft}{\mathsf{QFT}}

\newcommand{\ee}{\mathrm{e}}
\newcommand{\ii}{\mathrm{i}}

\newtheorem{mtheorem}{Theorem}
\newtheorem{mcorollary}[mtheorem]{Corollary}
\newtheorem{openproblem}{Open Problem}

\newtheorem{theorem}{Theorem}[section]
\newtheorem{lemma}[theorem]{Lemma}

\newtheorem{claim}[theorem]{Claim}
\newtheorem{corollary}[theorem]{Corollary}

\newtheorem{definition}[theorem]{Definition}

\newtheorem{remark}[theorem]{Remark}

\newtheorem*{theorem*}{Theorem}

\title{Quantum Worst-Case to Average-Case Reductions\\ for All Linear Problems}
\author{
Vahid R. Asadi\thanks{University of Waterloo. Email: \texttt{vrasadi@uwaterloo.ca}.}
\and
Alexander Golovnev\thanks{Georgetown University. Email: \texttt{alexgolovnev@gmail.com}.}
\and
Tom Gur\thanks{University of Warwick. Email: \texttt{tom.gur@warwick.ac.uk}. Tom Gur is supported by the UKRI Future Leaders Fellowship MR/S031545/1 and an EPSRC New Horizons Grant EP/X018180/1.}
\and
Igor Shinkar\thanks{Simon Fraser University. Email: \texttt{ishinkar@sfu.ca}.}
\and
Sathyawageeswar Subramanian\thanks{University of Warwick. Email: \texttt{Sathya.Subramanian@warwick.ac.uk}.}
}
\date{}

\begin{document}
\maketitle

\begin{abstract}
We study the problem of designing worst-case to average-case reductions for quantum algorithms. For all linear problems, we provide an explicit and efficient transformation of quantum algorithms that are only correct on a small (even sub-constant) fraction of their inputs into ones that are correct on \emph{all} inputs. This stands in contrast to the classical setting, where such results are only known for a small number of specific problems or restricted computational models. En route, we obtain a tight $\Omega(n^2)$ lower bound on the average-case quantum query complexity of the Matrix-Vector Multiplication problem. 

Our techniques strengthen and generalise the recently introduced \emph{additive combinatorics} framework for classical worst-case to average-case reductions (STOC 2022) to the quantum setting. We rely on quantum singular value transformations to construct quantum algorithms for linear verification in superposition and learning Bogolyubov subspaces from noisy quantum oracles. We use these tools to prove a quantum local correction lemma, which lies at the heart of our reductions, based on a noise-robust probabilistic generalisation of Bogolyubov's lemma from additive combinatorics.

\end{abstract}


\pagenumbering{roman}
\thispagestyle{empty}
\newpage
\setcounter{tocdepth}{2}
\tableofcontents
\newpage
\pagenumbering{arabic}

\newpage
\section{Introduction}
\label{sec:intro}
Average-case complexity is a central area of research in the theory of computing, which studies algorithms that solve problems on average inputs. This notion provides a paradigm for designing efficient algorithms that work on many relevant inputs, even if the worst-case complexity of the problem is high (cf., the standard textbooks \cite{goldreich_2008,arora2009computational}).

Worst-case to average-case reductions are transformations of algorithms that are correct on a~small fraction of their inputs into algorithms that are correct \emph{on all} inputs. That is, given an algorithm $\ALG$ for computing a function $f$ that satisfies $\Pr_x[\ALG(x)=f(x)] \geq \alpha$, the goal is to boost the \emph{success rate} $\alpha$ (i.e., the fraction of inputs upon which the algorithm is correct) to~$1$ without significantly increasing the algorithm’s complexity. We stress that for general problems even if there is an efficient way to verify the output of $\ALG$ (e.g., if $\ALG$ outputs a flag indicating that it has succeeded), it is still unclear if such a goal is achievable at all.

We can view worst-case to average-case reductions both as a means for deriving average-case hardness results from worst-case lower bounds, and also as a paradigm for designing worst-case algorithms by first constructing algorithms that are only required to succeed on a small fraction of their inputs and then using the reduction to obtain algorithms that are correct on all inputs.

In this work, we study \emph{efficient} worst-case to average-case reductions where the success rate is low (e.g., where the average-case algorithm is only correct on 1\% of the inputs\footnote{We stress that the 1\% regime is far more challenging. Indeed, in the 99\% regime, simple self-correction can be used to obtain fine-grained worst-case to average-case reductions for a number of problems \cite{BlumLR90}.}) in the setting of \emph{quantum computing}. On the one hand, the quantum setting is more complex and poses significant challenges since we need to transform a much larger class of average-case algorithms. On the other hand, the quantum setting also allows us to use powerful quantum procedures in the design of the worst-case algorithms to avoid classical bottlenecks.

This paper deals with the following fundamental question regarding quantum average-case complexity:
\begin{tcolorbox}
\begin{center}
    Is it possible to transform efficient quantum algorithms that are only correct on $1\%$ of their inputs into efficient quantum algorithms that are correct \emph{on all} inputs?
\end{center}
\end{tcolorbox}

\subsection{Our contributions}
We provide a strong, positive answer to the question above. In fact, we show that not only are such transformations possible, but that we can construct explicit and efficient quantum worst-case to average-case reductions \emph{for all linear problems} with only constant blowup in the complexity.

This stands in stark contrast to the case of classical algorithms, where such reductions are only known for a small number of specific problems or restricted models. Furthermore, our reduction not only supports average-case algorithms in the $1\%$ regime but also algorithms where the success rate $\alpha$ tends to zero; that is, algorithms that are only correct on a vanishing fraction of their inputs.


In the following, generalising the definition of $\mathbf{BQP}$ algorithms to fine-grained search problems in the standard way, we define an average-case quantum algorithm as a uniformly generated set of quantum circuits $\{C_n\}_{n \in \N}$ which, upon measurement, output correctly with probability $\alpha$, where the probability is taken over \emph{both} the random input and the measurements. (See formal definitions in \cref{sec:main}.)

A linear problem is characterised by a family of matrices $\mathcal{M}:=\{M_n\in\F^{n\times n}\}_{n\in\N}$, where on input $v\in\F^n$ the solution to the problem is the vector $Mv$, omitting the subscript on $M$ for readability. Linear problems constitute one of the most fundamental classes of problems, generalising many important computational tasks such as polynomial evaluation, computing discrete Fourier transformations, homology, and various computational tasks for error correcting codes.

Our main result is a worst-case to average-case reduction which shows that for all linear problems, an efficient quantum algorithm that is only successful on a small fraction of its inputs can be explicitly transformed into a similarly efficient quantum algorithm that is correct with high probability \emph{on all inputs}. $\Pi_x$ below denotes an orthogonal projection on the output register of $\ALG$ that represents measuring the outcome $x\in\F^n$ in the standard basis.

\begin{mtheorem}[Informally stated; see \cref{thm:main-circuit}]
\label{thm:circuit-informal}
    Let $\F$ be a finite field, $M:=\{M_n\in\F^{n\times n}\}_{n\in\N}$ be any linear problem, and $\ALG$ be an average-case quantum algorithm of (gate) complexity $T$ for $M$ satisfying 
    \begin{equation*}
        \Pr_{v,\ALG}[\ALG(v) = Mv]  =
        \E_{v\in\F^n}\left[\norm{\Pi_{Mv}{\ALG}\ket{v}\ket{0}}^2\right] \geq \alpha\;.
    \end{equation*}
    
    Then, for every constant $\delta>0$, there exists a worst-case quantum algorithm ${\ALG'}$ of (gate) complexity $(T+n^{3/2})\cdot\poly(1/\alpha)$ that succeeds over all inputs with high probability, i.e., 
    \begin{equation*}
    \forall v \in \F^n \quad
        \Pr_{\ALG'}[\ALG'(v) = Mv]  =
        \norm{\Pi_{Mv}{\ALG'}\ket{v}\ket{0}}^2 \geq 1 - \delta\;.
    \end{equation*}
\end{mtheorem}

Note that every linear problem can be trivially solved in time $O(n^2)$. Our result shows that any non-trivial (subquadratic) average-case quantum algorithm can be transformed into a non-trivial (sub-quadratic) quantum algorithm that works for \emph{all} inputs. We also remark that constructing worst-case to average-case reductions for linear problems becomes significantly harder as we consider smaller fields, as we discuss in the technical overview (\cref{sec:overview}). We stress that our reductions also hold for small fields, including $\F_2$.

Our proof of \cref{thm:circuit-informal} builds upon a machinery that we develop in the quantum query model. This allows us to obtain the following worst-case to average-case reduction for the fundamental and well-studied problem of Matrix-Vector Multiplication in the quantum query model \cite{BuhrmanSpalek04,KothariThesis}.

\begin{mtheorem}[Informally stated; see \cref{thm:main-query}]
\label{thm:query-informal}
    Let $\ALG$ be an average-case quantum query algorithm with oracle access to a matrix $M$ and a vector $v$ over a finite field $\F$. Suppose that $\ALG$ makes $q$ queries and satisfies
    \begin{equation*}
        \Pr_{\substack{M,v,\\ \ALG}}[\ALG^{M,v} = Mv]  =
        \E_{\substack{M \in \F^{n\times n}\\v\in\F^n}}\left[\norm{\Pi_{Mv}{\ALG}^{M,v}\ket{0}}^2\right] \geq \alpha\;.
    \end{equation*}
    Then, for every constant $\delta>0$, there exists a worst-case quantum query algorithm $\ALG'$ with query complexity $(q+n^{3/2})\cdot\poly(1/\alpha)$ that succeeds on \emph{all inputs} with high probability, i.e.,
    \begin{equation*}
    \forall M \in \F^{n\times n}, v \in \F^n \quad
        \Pr_{\ALG'}[(\ALG')^{M,v} = Mv]  =
        \norm{\Pi_{Mv}{(\ALG')}^{M,v}\ket{0}}^2 \geq 1 - \delta\;.
    \end{equation*}
\end{mtheorem}

In the query model, it is known that the worst-case quantum query complexity of Matrix-Vector Multiplication has a tight lower bound of $\Theta(n^2)$ (see, e.g., \cite{KothariThesis}). Hence as an immediate corollary of \cref{thm:query-informal}, we obtain a tight unconditional average-case lower bound for Matrix-Vector multiplication, showing that the problem remains hard even if the quantum algorithm is only required to succeed on a small fraction of the inputs.

\begin{mcorollary}
\label{cor:lowerbound}
For every constant $\alpha>0$, every average-case quantum query algorithm for Matrix-Vector Multiplication with success rate $\alpha$ must make $\Omega(n^2)$ queries. 
\end{mcorollary}

\subsection{Related work}
The study of the average-case complexity originates in the work of Levin~\cite{L86}, and follow-up works such as
\cite{BCGL92}. A long line of works established various barriers to designing worst-case to average-case reductions for $\mathbf{NP}$-complete problems (see, e.g., \cite{ImpagliazzoL90,I11} and references therein). We refer the reader to the classical surveys by Impagliazzo~\cite{I95}, Bogdanov and Trevisan~\cite{BT06}, and Goldreich~\cite{Goldreich2011} on this topic. 

On the other hand, there are known worst-case to average-case reductions for certain problems~\cite{L91,FF93,BFNW93,A1996,STV01}. For example, the problems underlying the classical number-theoretic cryptography (such as the RSA, discrete logarithm, and quadratic residuosity problems) are random self-reducible, and, therefore, admit efficient worst-case to average-case reductions (for fixed parameters). The celebrated work of Shor~\cite{S1994} gave a polynomial time quantum algorithm breaking the number-theoretic cryptosystems, which sparked interest in post-quantum cryptography, i.e., cryptography secure even against (polynomial time) quantum adversaries. A number of quantum and classical worst-case to average-case reductions~\cite{A1996,MR04,R09,LPR13,LS15,G10} allowed us to base the security of (post-quantum) lattice-based cryptography on the \emph{worst-case} quantum hardness assumptions for certain computational problems. Another interesting example of a (quantum) worst-case to average-case reduction was recently given in~\cite{LdW21} for problems related to phase estimation.

Recently, the study of fine-grained complexity~\cite{V18} of algorithmic problems sparked interest in designing \emph{efficient} worst-case to average-case reductions, i.e., reductions that do not suffer a polynomial overhead in the running time. Such reductions are often motivated by fine-grained cryptographic applications~\cite{BRSV17, BRSV18, GR18, LLW19, BBB21,DLV20}.  

In a recent work \cite{AGGS22}, a new framework for showing efficient worst-case to average-case reductions was introduced. This framework uses the quasi-polynomial Bogolyubov-Ruzsa lemma to show reductions in the classical setting that support the low-agreement regime, where the average-case algorithm is only guaranteed to succeed on 1\% of the inputs. This framework was used to obtain reductions for a few specific problems or restricted models of classical computation. 


\subsection{Open problems}
Our work opens up several new directions of investigation. Below, we highlight three open problems of particular interest. Recall that in \cref{thm:circuit-informal}, we have shown fine-grained \emph{quantum} worst-case to average-case reductions, where the success rate is arbitrarily small, \emph{for all linear problems}. These reductions make crucial use of quantum procedures that speed up classical computational tasks such as linear verification and learning of Bogolyubov subspaces from approximate indicators encoded in noisy quantum oracles.

Interestingly, unlike the general result above for quantum algorithms, in the classical setting, the aforementioned computational tasks constitute a complexity bottleneck, and in turn such worst-case to average-case reductions are only known for a small number of specific problems or restricted computational models \cite{AGGS22}. It remains open whether such general results for all linear problem can also be obtained in the classical setting.

\begin{openproblem}
    Are there efficient transformations of classical algorithms (or circuits) for general linear problems, which are only correct on 1\% of their inputs, into similarly efficient worst-case classical algorithms (or circuits)?
\end{openproblem}

Returning to the quantum setting, it is natural to ask whether the framework we constructed can be extended beyond the class of linear problems.

\begin{openproblem}
    Can efficient average-case quantum algorithms with success rate 1\% for large classes of \emph{non-linear} problems be transformed into efficient worst-case quantum algorithms?
\end{openproblem}

The last open problem we would like to raise refers to quantum algorithms that act on data of exponential size, encoded in the amplitudes of a quantum state, as in the celebrated HHL algorithm \cite{Harrow2009}. It would be highly appealing to extend the framework presented in this paper to this setting, where quantum algorithms are extremely powerful.

\begin{openproblem}
    Can we obtain fine-grained worst-case to average-case reductions for quantum algorithms such as the HHL algorithm, in which the data (and the output) are encoded in the amplitudes of a quantum state?
\end{openproblem}



\subsection*{Acknowledgements}
We thank Richard Cleve for discussions that inspired this work. We also thank the anonymous referees for their very helpful suggestions.

\subsection*{Organisation}
The rest of the paper is organised as follows. In \cref{sec:overview}, we provide an overview of our techniques and a high-level overview of the proof of \cref{thm:query-informal}. In \cref{sec:preliminaries}, we give the necessary background material. In \cref{sec:main}, we state a general technical lemma (\cref{lem:main}) and show how to use it to derive our worst-case to average-case reductions (\cref{thm:circuit-informal} and \cref{thm:query-informal}). The subsequent sections construct the components that are necessary to prove the aforementioned main technical lemma, as follows.

In \cref{sec:quantumtools}, we provide a toolkit of quantum algorithms for local correction; our techniques employ the quantum singular value transformation machinery, which we discuss in \cref{app:qsvt}. Then, in \cref{sec:AC}, we prove a robust and probabilistic version of Bogolyubov's lemma from additive combinatorics, and we use it together with our toolkit of quantum algorithms to obtain a quantum local correction lemma, which will play a key role in our worst-case to average-case reductions. Finally, in \cref{sec:reduction} we use all the tools above to prove the main technical lemma (\cref{lem:main}).

\section{Techniques}
\label{sec:overview}

We provide a technical overview, highlighting the main conceptual and technical ideas that we use. Our approach builds on the \emph{additive combinatorics} framework, recently introduced in \cite{AGGS22}; however, our setting, which captures all linear problems, is significantly more involved and presents non-trivial challenges that require new additive combinatorics techniques, complexity theoretic ideas, and quantum algorithms that are tailored to this setting.

We start in \cref{tech:approach}, where we present our high-level approach for transforming average-case quantum algorithms for linear problems into worst-case quantum algorithms. In \cref{tech:qcorrection}, we discuss the key technical component underlying our reductions, which is a quantum local correction lemma based on a robust, probabilistic version of Bogolyubov's lemma. Our quantum local correction lemma crucially relies on four quantum algorithms, which we outline in \cref{tech:qalgo}, that allow us to efficiently flag correct solutions in superposition, construct noisy oracles for approximate indicator functions, efficiently sample good inputs, and learn Bogolyubov subspaces from noisy set-indicating quantum oracles. Finally, in \cref{tech:reduction}, we discuss how to apply the quantum local correction lemma to obtain the desired worst-case to average-case reductions.

\subsection{An additive combinatorics approach}
\label{tech:approach}

In this overview, for simplicity of exposition, we will focus on the proof of \cref{thm:query-informal} in the quantum query model; we in fact optimise both the query and gate complexities, so that we can then show how to extend the result to (uniform) quantum circuits for all linear problems. 

Recall that the Matrix-Vector Multiplication problem ($\mathsf{MvM}$) is defined as follows.


\vspace{2ex}
\noindent{$\quad$\small\underline{\textsc{Matrix-Vector Multiplication} ($\mathsf{MvM}$)}} \\
\noindent\textbf{Input:} Oracle access to a matrix $M\in\F^{n\times n}$ and a vector $v\in\F^n$.\\
\noindent\textbf{Output:} The matrix-vector product $Mv$.\\

Let $\ALG$ be an \emph{average-case} quantum algorithm that is given oracle access to $M$ and $v$, defined by
\begin{equation*}
    U_M\ket{j, k, z} = \ket{j, k, z \oplus M_{jk}} \quad  \textrm{and} \quad U_v\ket{j, z} = \ket{j, z \oplus v_{j}},
\end{equation*}
for all indices $j,k\in [n]$ and $z\in\F$. Suppose that using $q$ queries to $M$ and $v$, the quantum algorithm $\ALG$ satisfies
    \begin{equation*}
        \Pr_{\substack{M,v,\\ \ALG}}[\ALG^{M,v} = Mv]  =
        \E_{\substack{M \in \F^{n\times n}\\v\in\F^n}}\left[\norm{\Pi_{Mv}{\ALG}^{M,v}\ket{0}}^2\right] \geq \alpha\;,
    \end{equation*}
where $\Pi_{Mv}$ is an orthogonal projection on the output register of $\ALG$ that indicates whether the algorithm outputs the correct answer.

We would like to explicitly transform $\ALG$ into a \emph{worst-case} quantum algorithm $\ALG'$ that computes $Mv$ with high probability for \emph{every} matrix $M$ and vector $v \in \F^n$; that is, show that for every constant $\delta>0$, there exists a worst-case quantum algorithm ${\ALG'}$ satisfying
    \begin{equation*}
    \forall M\in\F^{n\times n}, v \in \F^n \quad
    \Pr_{\ALG'}[(\ALG')^{M,v} = Mv]  =
        \norm{\Pi_{Mv}{(\ALG')}^{M,v}\ket{0}}^2 \geq 1 - \delta\;.
    \end{equation*}
To simplify the discussion, unless specified otherwise, in this overview we restrict our attention to the field $\F_2$, and to arbitrarily small constant values of the success rate parameter $\alpha>0$ of average-case algorithms (say $\alpha=0.01$).

\paragraph{Worst-case to average-case reductions via additive combinatorics.}
Our starting point is the additive combinatorics framework that was recently introduced in \cite{AGGS22}. Namely, we'd like to decompose each input vector $v$ into a sum of vectors on each of which the average-case algorithm $\ALG$ is correct. However, since a simple linear decomposition of each $v \in \F^n$ into correctly-computed inputs does not exist,\footnote{Indeed, as in the classical setting, consider the simple counterexample where the average-case algorithm $\ALG^{M,v}$ outputs $M \cdot v$ in case that the first coordinate of $v$ is~$1$ and otherwise outputs $0$. Note that in this case the success rate is $\alpha \geq 1/2$, yet no linear decomposition could self-correct matrix-vector multiplication where the first coordinate of $v$ is~$1$. Indeed, any such decomposition $v=\sum_i v_i$ would have a $v_i$ with the first element $1$, where $\ALG^{M,v}$ fails.} we shall need more involved machinery from the field of additive combinatorics.

To this end, we will start with Bogolyubov's lemma, a fundamental result in additive combinatorics which shows that the $4$-ary sumset of any dense set in $\F_2^n$ contains a large linear subspace. More accurately, recall that the sumset of a set $A$ is defined as $A+A = \{ a_1+a_2 \,:\, a_1,a_2 \in A \}$, and similarly $4A = \{ a_1+a_2+a_3+a_4\,:\, a_1,a_2,a_3,a_4 \in A \}$. These objects can be thought of as a combinatorial analogue of an approximate subgroup. Considering these sumsets allows us to extract subspace structure out of an unstructured set, as encapsulated in the following lemma.
\begin{lemma}[Bogolyubov's lemma]
    For any subset $A \subseteq \F_2^n$ of density $|A|/2^n \geq \alpha$, there exists a subspace $V \subseteq 4A$ of dimension at least $n - \alpha^{-2}$.
\end{lemma}

To see the initial intuition for the additive combinatorics approach to designing worst-case to average-case reductions, we first make the simplifying assumption that the average-case algorithm $\ALG$ receives a ``good'' matrix $M \in \F^{n \times n}$ for which it is successful with probability $\alpha$, taken over the measurement and the random input vector $v \in \F^n$; that is, 
    \begin{equation*}
        \Pr_{v}[\ALG^{M,v} = Mv]  =
        \E_{v\in\F^n}\left[\norm{\Pi_{Mv}{\ALG}^{M,v}\ket{0}}^2\right] \geq \alpha\;.
    \end{equation*}
In \cref{tech:reduction}, we will show how to extend our approach to work on average-case matrices as well. 

First note that by an averaging argument, there exists a set of size $(\alpha \cdot 2^n)/2$ of input vectors $v \in \F^n$ on which $\ALG$ correctly computes the output with probability at least $\alpha/2$; denote this set by $X$. Next, observe that Bogolyubov's lemma shows that there exists a large subspace~$V$ such that every $v \in V$ can be decomposed as
\begin{equation}
    \label{eq:decomp}
    v = x_1 + x_2 + x_3 + x_4, \textrm{ for } x_1,x_2,x_3,x_4 \in X \;.
\end{equation}
 Recall that each $x_i \in X$ can be computed correctly by the average-case algorithm with probability at least $\alpha/2$. This suggests the natural approach of locally correcting each $v \in V$ using four inputs upon which the average-case algorithm has a non-negligible success probability. 
 
\paragraph{The challenges.} While the discussion above outlines a promising approach, already at this point there are substantial difficulties that arise when trying to pursue it. For starters, even on good inputs in $X$ the average-case algorithm only computes correctly with probability $\alpha/2$ (which could tend to zero!), and so it is unclear how to amplify the success probability of the algorithm such that it computes correctly on all four elements of $X$ in the decomposition in \cref{eq:decomp} with high probability, which is necessary for the self-correction of inputs in the Bogolyubov subspace $V$.

We stress that \emph{this is not the case} with the problems that were considered in \cite{AGGS22} (where verification was done using Freivald's algorithm for matrix-matrix multiplication, or using the unbounded preprocessing power in the data structure model). In contrast, in the setting of general linear problems, naive verification of a matrix-vector product costs $O(n^2)$, which completely trivialises the problem; indeed, in the classical setting, $o(n^2)$-query verification of matrix-vector products is impossible. Further issues include the fact that Bogolyubov's lemma only guarantees the existence of a decomposition into elements of $X$, whereas we need to explicitly obtain such a decomposition, as well as the fact that the argument above only holds for correcting inputs that are \emph{inside the Bogolyubov subspace} $V$, whereas we need to locally correct \emph{all inputs}.

Unfortunately, the classical framework that was shown in \cite{AGGS22} fails to address the aforementioned problems in the setting of general linear problems. Indeed, it is not clear whether the ambitious task of constructing worst-case to average-case reductions for all linear problems is at all possible for classical algorithms.

Instead, our approach for overcoming these challenges is \emph{inherently quantum}, and it would require more involved ideas, both technically and conceptually. We will show a new, noise-robust version of Bogolyubov's lemma (see \cref{tech:qcorrection}), and together with a toolkit of quantum algorithms (building on quantum singular value transformations) that we develop (see \cref{tech:qalgo}), we will prove a new quantum local correction lemma that will play a key role in our reductions, which are tailored to the strengths and limitations of the quantum setting.


\subsection{Quantum local correction via a robust Bogolyubov lemma}
\label{tech:qcorrection}


Our starting point for addressing the difficulties outlined in \cref{tech:approach} is a new \emph{robust} and \emph{probabilistic} analogue of Bogolyubov's lemma. These additional structural properties will in turn enable us to deal with inputs outside of the Bogolyubov subspace and efficiently obtain explicit decompositions of inputs inside the Bogolyubov subspace into inputs upon which the average-case algorithm is correct.

Then, using the robust probabilistic Bogolyubov lemma, we will show a local correction lemma for quantum algorithms for linear problems. Our quantum local correction lemma will crucially make use of the robustness and probabilistic properties, as well as the four quantum procedures that we present in \cref{tech:qalgo}. 


\paragraph{A robust probabilistic Bogolyubov lemma.} Recall that our high-level idea for locally correcting faulty inputs in the subspace $V$ guaranteed by Bogolyubov's lemma proceeds by decomposing vectors $v \in V$ into a linear combination of \textit{good inputs} from the set $X$. However, to prove our quantum local correction lemma, we need to be able to: (1) deal with vectors  \emph{outside of the Bogolyubov subspace $V$}, and (2) efficiently obtain an explicit decomposition of vectors inside $V$ into good inputs in $X$. To this end, we shall first need to strengthen Bogolyubov's lemma to obtain the following structural properties.


\begin{itemize}
    \item \emph{Robustness.} Our local correction lemma relies on learning the heavy part of the Fourier spectrum of a noisy representation of an approximate indicator of the set $X$ in superposition, using a quantum procedure that we discuss in \cref{tech:qalgo}. In this setting, we can only learn the heavy Fourier coefficients of a function $g$ satisfying
    \begin{equation*}
        \forall S\subset[n] \quad \abs{|\widehat{1_X}(S)|^2 - |\hat{g}(S)|^2} \leq \epsilon \;,
    \end{equation*}
    for a sufficiently small $\eps$. In fact, each Fourier coefficient we obtain may come from a \emph{different function} $g$. To deal with this problem, which is \emph{unique to the quantum setting}, we need a version of Bogolyubov's lemma where the subspace is defined via a robust set of linear constraints that can accommodate for the noisy Fourier spectrum.
    
    
    \item \emph{Probabilistic decomposition.} Bogolyubov's lemma shows the \emph{existence} of a large subspace $V \subset 4X$, admitting a linear decomposition of each $v \in V$ into four vectors from the set of good inputs $X$ upon which the average-case algorithm is correct. However, we need to efficiently obtain an explicit decomposition $v = x_1 + x_2 + x_3 + x_4$, where each $x_i \in X$, for each vector $v \in V$. Hence, we show that each vector admits sufficiently many such decompositions such that one can be efficiently sampled using a quantum sampling procedure that we show in \cref{tech:qalgo}.
\end{itemize}

We thus prove the following \emph{robust probabilistic Bogolyubov lemma}, which will allow us to deal with the challenges above. Given an unstructured set $X \seq \F^n$, denote the heavy part, exceeding a threshold $\gamma$, of the Fourier spectrum of $X$ by $\Spec_X(\gamma) = \{r \in \F^n \setminus \{0\} : \abs{\hat{1}_X(r)} \geq \gamma\}$.

\begin{lemma}[Informally stated, see \cref{lem:apx-bogolyubov}]
\label{lem:inf-Bogolyubov}
   For every $X \seq \F^n$ of density at least $\alpha$, let $R \seq \F^n$ be a set such that
    $\Spec_X(\alpha^{3/2}) \seq R \seq \Spec_X(\frac{\alpha^{3/2}}{2})$.
    Let $V = \{v \in \F^n : \ip{v,r} = 0 \ \forall r \in R\}$.
    Then $\dim(V) \geq n-O(\alpha^{-2})$, and for all $v \in V$ it holds that
    \begin{equation*}
        \Pr_{x_1,x_2,x_3 \in X}[v - x_1 - x_2 - x_3 \in X] \geq \alpha^2 \;.
    \end{equation*}
\end{lemma}
We stress that the probability in \cref{lem:inf-Bogolyubov} is taken over $x_1,x_2,x_3$ that are uniformly sampled from $X$ (rather than $\F^n$). To efficiently obtain such a decomposition, we can sample $x_1,x_2,x_3$ using the quantum sampling procedure shown in \cref{tech:qalgo}.

\paragraph{Locally correcting outside of the Bogolyubov subspace.} While the robust probabilistic Bogolyubov lemma allows us to locally correct inputs inside the subspace $V \subseteq 4X$, our goal is to obtain a worst-case quantum algorithm, hence we need to be able to handle any vector $v \in \F^n$, and not just those in $V$.

Towards this end, we would like to shift each vector $v \in \F^n \setminus V$ into the Bogolyubov subspace. Using the robustness property of \cref{lem:inf-Bogolyubov} and the quantum algorithm for learning Bogolyubov subspaces from a noisy representation of indicator functions in superposition (see \cref{tech:qalgo}), we further decompose each vector in $\F^n$ into a sum of elements in the subspace $V$ and a \emph{sparse} shift-vector~$s$, which can then be corrected efficiently due to its sparsity.

In more detail, let $R \subseteq \F^n \setminus \{0\}$ and $V = \{v \in \F^n : \ip{v,r} = 0 \; \forall r \in R\}$. We observe that there exists a collection of $t\leq \abs{R}$ vectors $B=\{b_1,\dots,b_t\},\,b_i \in \F^n$ and indices $k_1,\dots, k_t \in [n]$ such that $\sp(B) = \sp(R)$ and every vector $y \in \F^n$ can be written as $y = v + s$, where $v \in V$ and $s = \sum_{j=1}^t c_j \cdot \vec{e}_{k_j}$ for $c_j = \ip{y, b_j}$ and $\vec{e}_{k_j}$ is a unit vector. We emphasize that the sparsity of the decomposition is critical as it allows us to shift arbitrary vectors into the Bogolyubov subspace $V$ without unfavourably blowing up the complexity.

To efficiently obtain the basis $b_1,\dots,b_t \in \F^n$ and indices $k_1,\dots, k_t \in [n]$, we rely on the quantum algorithm for learning Bogolyubov subspaces described in \cref{tech:qalgo}, which allows us to obtain the approximate high Fourier coefficients of an approximation of $1_X$ from a noisy representation in superposition, within the desired complexity. This, in turn, allows us to compute the required basis and set of indices. 


\paragraph{The quantum local correction lemma.}
Putting together all of the components above, we can now state our quantum local correction lemma, which builds upon the robust probabilistic version of Bogolyubov's lemma and the quantum procedures that we will present in \cref{tech:qalgo}.

Loosely speaking, our local correction lemma allows us to efficiently obtain an explicit decomposition of any vector $v\in\F^n$ into a linear combination of the form
\begin{equation*}
    v = x_1 + x_2 + x_3 + x_4 + s
    \;,
\end{equation*}
where $x_1,x_2,x_3,x_4 \in X$ and $s \in \F^n$ is a \emph{sparse} vector.

Specifically, in \cref{tech:qalgo}, we will construct a noisy quantum oracle for approximating the indicator $1_X(v)$, quantum algorithm for uniformly sampling from~$X$, a quantum verification procedure $\widetilde{O}_X$ that for each $v \in \F^n$ computes the indicator $1_X(v)$ correctly with high probability, and a quantum algorithm that learns the Bogolyubov subspace from the noisy oracle that approximates $1_X$. By combining these four ingredients we prove the following result. 

\begin{lemma}[informally stated, see \cref{lem:quantum-local-correction}]
\label{lem:inf-correction}
 For a field $\F = \F_p$ and $\alpha$-dense set $X \seq \F^n$, there exists $t \leq 4/\alpha^2$ vectors $b_1,\dots,b_t \in \F_2^n$ and indices $k_1,\dots, k_t \in [n]$ satisfying the following. Given a vector $y \in \F^n$, define $s = \sum_{j=1}^t \ip{y, b_j} \cdot \vec{e}_{k_j}$, where $(\vec{e}_i)_{i \in [n]}$ is the standard basis. Then,
            \begin{equation*}
                \Pr_{x_1,x_2,x_3 \in X}[x_4 \in X] \geq \alpha^2
                \enspace,
            \end{equation*}
where $x_4$ is the vector such that $y =s + x_1 + x_2 + x_3 + x_4$.

    
    
    Furthermore, there exists a quantum algorithm that calls the quantum verification procedure $\widetilde{O}_X$ for $O(\log(1/\delta) \cdot (1/\alpha^{5}))$ times, uses $n \cdot \poly(1/\alpha) \cdot\poly\log |\F|$ additional elementary gates, and with probability at least $1-\delta$ returns the vectors $b_1,\dots,b_t \in \F^n$ and indices $k_1,\dots, k_t \in [n]$.
\end{lemma}

This quantum local correction lemma forms the cornerstone of our quantum average-case to worst-case reductions, as we explain in \cref{tech:reduction}. However, before we can turn to the details of the reduction, we need to address three gaps that still remain: How can we efficiently verify Matrix-Vector products, so that we can amplify the success probability of quantum algorithms that are only correct with probability $\alpha$ (which could tend to zero)? Can we efficiently sample from the set of good inputs $X$? How can we learn the Bogolyubov subspace from a noisy quantum oracle that encodes an approximation of the indicator $1_X$?

\subsection{A toolkit of quantum algorithms for local correction}
\label{tech:qalgo}
Next, we address the foregoing questions by presenting the key quantum procedures that are required for our quantum local correction lemma. We stress that these solutions employ the power of quantum algorithms, and indeed, the lack of such procedures in the classical setting is a bottleneck that blocks the path to general results such as \cref{thm:circuit-informal} and \cref{thm:query-informal} for \emph{classical} algorithms.

We achieve quantum speedups for the following four tasks, which are required for our quantum local correction lemma:
\begin{enumerate}
    \item Flagging correct matrix-vector products in superposition.
    \item Constructing noisy quantum oracles approximating the indicator function $\mathsf{1}_X$.
    \item Sampling of vectors from the set $X$ of good vectors.
    \item Learning Bogolyubov subspaces from noisy quantum oracles.
\end{enumerate}
Recall that it is essential to perform the above tasks in complexity $o(n^2)$; indeed our quantum algorithms make at most $O(n^{3/2})$ queries and use at most $\widetilde{O}(n^{3/2})$ additional one-qubit and two-qubit gates. We now proceed to briefly describe the above quantum procedures, and offer a quick glimpse at our techniques for obtaining them. See \cref{sec:quantumtools} for details.

\paragraph{Flagging correct matrix-vector products in superposition.} The first tool we shall need is a quantum subroutine for verifying whether a vector $b\in\F^n$ output by a quantum algorithm $\ALG$ is in fact the correct matrix-vector product $Mv$, a task that classically requires $\Omega(n^2)$ queries. This will play an important role in both amplifying the success probability of our local correction lemma (see \cref{tech:qcorrection}) and in the other quantum algorithms in the toolkit we develop.

Since it is crucial for our reduction to use such a verification procedure as a unitary subroutine in other quantum procedures, we unfortunately cannot apply existing verification algorithms from the literature. Furthermore, since the algorithm $\ALG$ only succeeds \emph{on average} with low probability, different input vectors have significantly different success probabilities. To address this, we give a procedure that flags all of the computational basis states in a superposition as either right or wrong in a way that respects the success distribution of $\ALG$, and then boost the amplitude of the solutions.

However, since the success probabilities of different inputs have high variance, choosing a fixed number of iterations of amplitude amplification causes problems of over-shooting and under-shooting. Instead, we apply a delicate argument involving fixed-point amplitude amplification, which converges monotonically towards the flagged state. 
Towards this end, we identify that the flagging operation above has the form of a \textit{block encoding}.\footnote{A unitary $U$ with the form $\begin{pmatrix}H& \cdot\\ \cdot & \cdot\end{pmatrix}$, which encodes another (subnormalised) matrix $H$ in its upper left block. } Having such an object enables us to use powerful techniques from the repertoire of quantum singular value transformations.

Implementing the strategy above, we construct the following quantum procedure for verifying matrix-vector products, which works \emph{in superposition} and marks an ancillary qubit attached to the output state of $\ALG$ whenever the state has non-zero overlap with the correct matrix-vector product $\ket{Mv}$. See \cref{sec:quantumtools-verification} for details.

\begin{lemma}[Informally stated; see \cref{lem:qVerification}]
\label{lem:inf-verification}
    Given access to a unitary oracle for a matrix $M\in\F^{n\times n}$, and a quantum algorithm $\ALG$ that takes as input a vector $v\in\F^n$ and produces as output a state that consists of a superposition over vectors in $\F^n$, 
    \[
        \ALG\ket{v,0} = \sum_{z\in\F^n}\gamma_z^v\ket{v}\ket{z}\ket{\mathrm{w(z,v)}}\;,
    \]
    there is a quantum algorithm $\ALG_{\mathsf{verified}}$ that annotates the output state of $\ALG$ with a flag marking the vector $Mv$ as correct, and marking all other vectors $z\neq Mv$ as incorrect with high probability, i.e.\
    \[
        \ALG_{\mathsf{verified}}\ket{v,0} \approx \gamma_*^v\ket{v}\ket{Mv}\ket{\mathrm{w(Mv,v)}}\ket{1}^{\mathrm{flag}}+\left(\sum_{z\neq Mv}\gamma_z^v\ket{v}\ket{z}\ket{\mathrm{w(z,v)}}\right)\ket{0}^{\mathrm{flag}}\;.
    \] $\ALG_{\mathsf{verified}}$ makes $O(1)$ uses of $\ALG$, $O(n^{3/2})$ queries to $U_M$ and $U_M^{\dagger}$, $\widetilde{O}(n)$ ancillary qubits, and $\widetilde{O}(n^{3/2})$ additional one-qubit and two-qubit gates.
\end{lemma}


\paragraph{Noisy quantum oracles approximating the indicator function $\mathsf{1}_X$.}
A key ingredient in our reductions is the indicator function or membership oracle $\mathsf{1}_X$ for the subset of good inputs exceeding probability threshold $\tau$, given by 
\begin{equation*}
    X_{\tau}\colon= \left\{v \in \F^n \colon \Pr[\ALG^M(v)=Mv] \geq {\tau} \right\} \;.
\end{equation*}
Here the notation $\ALG^M$ emphasises the fact that $\ALG$ has query access to the matrix $M$. In constructing this indicator function we are faced with yet another challenge: while fixed-point amplitude amplification is able to boost the probability of outputting the correct success/failure flag, it is insufficient to perform the type of thresholding operation that is required to implement the indicator function. The key difficulty is the fact that we need to \textit{simultaneously} mark inputs inside set $X_{\tau}$ with the flag one, and those outside $X_{\tau}$ with the flag zero (with high probability).

To overcome this issue, we note that $\ALG_{\mathrm{verified}}$ is also a block encoding (of another matrix), and by combining it with the heavier machinery of \emph{quantum singular value threshold projection}, we are able to obtain a noisy quantum oracle for a polynomial approximation of the indicator function on the set $X_{\tau}$ of good inputs to $\ALG$.
In \cref{sec:indicator}, we show how the technique of quantum singular value transformations can be used to construct a singular value threshold projection built on $\ALG_{\mathsf{verified}}$, which acts as a noisy indicator function on the set of inputs where $\ALG$ succeeds with high probability.

\paragraph{Sampling vectors from the set of good inputs.}
Using the ability to approximately flag correct and incorrect answers in superposition, we can synthesize a uniform superposition over the set $X_{\tau}$ that consists of all inputs on which the average-case algorithm computes correctly with probability at least $\tau$. Hence by preparing this state and measuring it, we obtain the following quantum speedup for sampling from $X_{\tau}$. This will allow us to boost the performance of our quantum local correction lemma (see \cref{tech:qcorrection}), that needs to sample elements from $X_{\tau}$ for a value of $\tau=\Theta(\alpha)$ in order to correct vectors in the Bogolyubov subspace. 


\begin{lemma}[Informally stated; see \cref{sec:quantum-sampling}]
\label{lem:inf-sampling}
    Given an average-case quantum algorithm $\ALG$ with success rate $\alpha$, there is a procedure $\mathsf{Q_{samp}}$ which with high probability produces a vector $v$ on which $\ALG$ succeeds with probability at least $\alpha/2$. $\mathsf{Q_{samp}}$ uses $\ALG$ $O(\frac{1}{\alpha})$ times, in addition to $\widetilde{O}(n^{3/2})$ additional one-qubit and two-qubit gates.
\end{lemma}

\paragraph{Learning Bogolyubov subspaces from noisy quantum oracles.}
In order to locally correct vectors that are outside of the Bogolyubov subspace $V$, we shall need to learn an approximation of $V$, so that we can shift arbitrary vectors into $V$ via sparse vectors (see \cref{tech:qcorrection}). 

Standard quantum procedures for sampling characters according to the probabilities defined by the Fourier spectrum of a function $f$ use a unitary oracle for $f$ in superposition, and are built on the Bernstein-Vazirani algorithm. However, in our setting we need to list-decode the heavy Fourier coefficients of a probabilistic implementation of an indicator function that we obtain from the linear verification procedure in \cref{lem:inf-verification}. Hence, we face the challenge of learning the heavy Fourier characters from a noisy implementation of the function. To this end, we generalise the techniques used to prove Adcock and Cleve's quantum Goldreich-Levin lemma \cite{AdcockCleve02}, and obtain the following quantum procedure.

\begin{lemma}[Informally stated; see \cref{lem:qGoldLev}]
\label{lem:inf-fourier}
    For a function $f:\F^n\to\F_2$, given a noisy quantum oracle $U_f$ acting on $m$ qubits such that for every input $x$, measuring the output qubit of $U_f$ produces $f(x)$ with probability at least $1-\epsilon$, there is a quantum procedure $\mathsf{C_{GL}}$ which produces outputs $y\in\F^n$ with probability $p_y$ that satisfy $\abs{p_y-|\hat{f}(y)|^2}\leq 4\epsilon$.  $\mathsf{C_{GL}}$ uses $U_f$ and $U_f^{\dagger}$ once, in addition to $O(n\log |\F|)$ additional one-qubit and two-qubit gates.
\end{lemma}
We stress that due to the noisy implementation, we do not obtain the heavy Fourier characters of the actual indicator function of the set of good inputs $X$, but rather an approximation of these Fourier coefficients. Fortunately, this guarantee suffices for the noise-robust Bogolyubov lemma that we presented in~\cref{tech:qcorrection}.

In fact, \cref{lem:inf-fourier} is not the last link in the chain: a further complication arises from the fact that any polynomial approximation of the indicator function $\mathsf{1}_X$ can only work well in a subset of $X$ and a subset of its complement, necessarily oscillating in a ``wasteland'' region corresponding to vectors on which the success probability of $\ALG$ lies in some range $(\tau-\delta, \tau+\delta)$.

The polynomial approximation happens at the level of a real-valued function (the success probability of $\ALG$ on input vectors) and the corresponding wasteland slice is apparently small, namely, an interval of length $2\delta$. However this hides the fact that the actual number of input vectors falling into this intermediate set can be alarmingly large, with density as high as $\frac{1-\alpha}{1-\alpha/2}$.

We show that in spite of this difficulty, a combination of careful error analysis of our quantum Bogolyubov subspace learning technique and, as we discuss in \cref{tech:reduction}, a carefully chosen random selection of the threshold $\tau$, allows us to efficiently learn the Bogolyubov subspace.

\subsection{Quantum worst-case to average-case reductions}
\label{tech:reduction}
    We first present a high-level overview of our reduction for the Matrix-Vector Multiplication problem in a simplified setting, then sketch how to extend the proof to obtain \cref{thm:circuit-informal} and \cref{thm:query-informal}.
    
    Recall that we start with an \emph{average-case} quantum algorithm $\ALG$ that is correct with probability $\alpha$ on a randomly chosen input, i.e.,
    \begin{equation*}
        \Pr_{\substack{M,v\\ \ALG}}[\ALG^{M,v} = Mv]  =
        \E_{\substack{M \in \F^{n\times n}\\v\in\F^n}}\left[\norm{\Pi_{Mv}{\ALG}^{M,v}\ket{0}}^2\right] \geq \alpha\;.
    \end{equation*}
    Our goal is to boost the success rate of the algorithm such that we obtain a \emph{worst-case} quantum algorithm that succeeds with high probability \emph{on all} inputs.
    
    For the purpose of a clear exposition, we first make the following simplifying assumptions: (1) we work over $\F=\F_2$, (2) we assume $\alpha>0$ is an absolute constant (and in turn, we will not optimise the dependency on it), (3) we fix the error parameter to an arbitrarily small constant, and (4) we take the average case only on the vectors $v \in \F^n$. In other words, we only consider the case where the given matrix $M$ is a \emph{good} matrix. Indeed, for the case that this does not hold, we will later present matrix self-correction techniques to ensure such a condition is satisfied with high probability.

    In order to describe the reduction the following notation will be convenient.
    For each $v \in \F^n$ let $p_v = \Pr_{\ALG}[\ALG(v) = Mv]$ be the probability that $\ALG$ correctly computes the output on input $v$, where the probability is only over the measurement of $\ALG$.

First, we define threshold sets as follows.
\begin{equation*}
        X_\tau = \{v \in \F^n : p_v > \tau\}\;.
\end{equation*}
Since $\Pr_{v \in \F^n}[\ALG(v) = Mv] \geq \alpha$, it implies that $\E_{v\in\F^n}[p_v] \geq \alpha$, and hence,
by Markov's inequality, for $\tau \leq \alpha/2$ it holds that $\frac{\abs{X_\tau}}{\abs{\F^n}} \geq \alpha/2$. Next, before describing the high-level overview of our proof, we first discuss a simplified warm-up that contains the main idea, while skipping some technical complications that we discuss later.

\paragraph{Warm-up.}
Ideally, we could construct the worst-case quantum algorithm $\ALG'$ as follows. We first learn all of the significant Fourier coefficients of $X_{\alpha/2}$ using the quantum learning procedure in \cref{lem:inf-fourier}. Then, by the quantum local correction lemma (i.e., \cref{lem:inf-correction}, which in turn, relies on the robust Bogolyubov lemma), these Fourier coefficients can be used to compute a decomposition of any input $v \in \F^n$ as 
\begin{equation*}
    v = s+x_1+x_2+x_3+x_4, \textrm{ where } x_1,x_2,x_3,x_4 \in X_{\alpha/2}\;,
\end{equation*}
and $s \in \F^n$ is a sparse vector that has only $O(1)$-non zero entries, which allows us to shift arbitrary inputs into the Bogolyubov subspace.

More specifically, for each input $v \in \F^n$, we can use the Fourier coefficients of $X_{\alpha/2}$ in order to compute the sparse vector $s$.
Then, we sample $x_1,x_2,x_3$ from $X_{{\alpha/2}}$, and set $x_4 = v - x_1-x_2-x_3 - s$.
By the quantum local correction lemma (\cref{lem:inf-correction}) we have that $\Pr[x_4 \in X_{\alpha/2}] \geq \poly(\alpha)$, and hence set the algorithm $\ALG'$ to compute
\begin{equation*}
    Ms + \ALG^M(x_1) + \ALG^M(x_2) + \ALG^M(x_3) + \ALG^M(x_4)\;.
\end{equation*}
Note that computing $Ms$ can be done in $O(n)$ time, where $O(\cdot)$ hides the sparsity of $s$, 
and $Mx_i$ can be computed correctly using $\ALG$ for each $i\in \{1,2,3,4\}$ with probability $\poly(\alpha)$. In total, the above procedure outputs the correct answer with probability at least $\poly(\alpha)$. By repeating the procedure $\poly(1/\alpha)$ times and verifying the result at each time, we obtain a reduction that succeeds with high probability.

\medskip
\paragraph{Random thresholds.}
The actual reduction is more subtle. The main reason for this is that the lemma we use to learn the Bogolyubov subspace (\cref{lem:inf-fourier}) cannot sample from the Fourier spectrum of the set $X_{\alpha/2}$ exactly, as computing $p_v$ exactly for a particular $v$ requires a large number of samples, and so, we cannot apply \cref{lem:inf-correction} directly on the set $X_{\alpha/2}$.

Instead, we choose a parameter $\tau \in [\alpha/4, \alpha/2]$, and consider the sets $X_{\tau}$ and $X_{\tau'}$, where $\tau' = \tau - O(\alpha^{3/2})$.
We choose the parameters so that 
\begin{equation*}
    \abs{\frac{\abs{X_{\tau'}}}{\abs{\F^n}} - \frac{\abs{X_{\tau}}}{\abs{\F^n}}} = O(\alpha^{3/2})\;.
\end{equation*}
See \cref{cor:close-fourier-coeffs} for details on how the random threshold $\tau$ is chosen. Then, we apply \cref{lem:inf-correction} with respect to \emph{some} set $X^*$ such that with high probability $X_{\tau} \seq X^* \seq X_{\tau'}$. Note that we do not have any structural guarantee about $X^*$ containing a particular vector in $X_{\tau'} \setminus X_{\tau}$.

We use \cref{lem:inf-fourier} in order to obtain all significant Fourier coefficients of $X^*$. Observing that $\frac{\abs{X^* \triangle X_{\tau}}}{\abs{\F^n}} = O(\alpha^{3/2})$,
it follows that 
\begin{equation*}
    \abs{\hat{X^*(y)} - \hat{X_\tau(y)}} = O(\alpha^{3/2})\;,
\end{equation*}
for all $y \in \F^n$, and hence we can use the Fourier coefficients of $X^*$ to approximate the  Fourier coefficients of $X_{\tau}$. Here by $\hat{X}$ we mean the Fourier coefficients of the indicator function of the set $X$.

By the discussion above, we may assume that we know a good approximation of all significant Fourier coefficients of $X_{\tau}$.
The rest of the reduction follows the same plan as described in the warm-up above.

\paragraph{Extending the average-case over both matrices and vectors.} So far, we obtained a worst-case to average-case reduction where the average-case condition only refers to the vectors in the Matrix-Vector Multiplication problem. To extend the average case to both vectors and matrices, we need an additional layer of matrix local correction. To this end, we use the technique of shifting the given matrix~$M$ by a random matrix~$R$, which with probability $\Omega(\alpha)$ shifts the input to the set of good matrices, where  matrix-vector multiplications are computed correctly for an $\Omega(\alpha)$-fraction of vectors. See \cref{sec:main-mvm} for details. 

\paragraph{Quantum algorithms for linear problems.}
Throughout the argument that we outlined above, we employed a proof strategy that only invokes gate-efficient quantum algorithms (i.e.,\ with gate complexities that are larger than the query complexity by at most a polylog factor). We do this precisely due to our interest in uniform quantum algorithms for linear problems: by a careful instantiation of families of quantum circuits as quantum query algorithms that match the setting of \cref{thm:main-query}, we obtain in \cref{thm:main-circuit} quantum worst-case to average-case reductions for all linear problems. See \cref{sec:main-all} for details.

\section{Preliminaries}\label{sec:preliminaries}

We establish some minimal standard preliminaries and notation regarding quantum algorithms, and refer the reader to standard textbooks such as \cite{Nielsen2010} for details.

\subsection{Quantum unitary oracles and Fourier transforms}

We start by providing basic notation and definitions regarding discrete and quantum Fourier transforms, as well as quantum oracles.

\paragraph{Discrete and quantum Fourier transforms.}
Let $\F=\F_p$ be the prime field of size~$p$. For a function $f:\F^n\to\R$, the Fourier coefficient $\widehat{f}(y)$ for $y\in\F^n$ representing characters $\chi_y$ is given by

\begin{equation}
\label{eq:fourier-coefts}
    \widehat{f}(y) = \frac{1}{p^n}\sum_{x\in\F^n}\omega^{x\cdot y}f(x) \; ,
\end{equation}
where $\omega=\ee^{2\pi\ii/p}$ is a primitive $p^{th}$ root of unity.

Denote the quantum Fourier transform over $\F$ by $\qft_p$, having the action
\begin{equation}
    \qft_p\ket{x} = \frac{1}{\sqrt{p}}\sum_{y\in\F_p}\omega^{x\cdot y}\ket{y} \; .
\end{equation}

\paragraph{Quantum unitary oracles.}
We let quantum algorithms access a matrix $M\in \F^{n\times n}$ via a unitary oracle $U_M$ that performs the map
\begin{equation}
\label{eq:quant-input-oracle-mat}
    U_M\ket{j, k, z} = \ket{j, k, z \oplus M_{jk}} \;,
\end{equation}
for all indices $j,k\in [n]$ and $z\in\F$, where $\oplus$ denotes addition over~$\F$. For vectors $v\in\F^n$, oracle access means a unitary $U_v$ that returns components of the vector when queried with an index, analogously to \cref{eq:quant-input-oracle-mat}, i.e.,
\begin{equation}
\label{eq:quant-input-oracle-vec1}
    U_v\ket{j, z} = \ket{j,  z \oplus v_{j}}\;,
\end{equation}
for all $j\in[n]$ and $z\in\F$. 
Unless otherwise stated, whenever we assume access to a unitary $U$ as an oracle we also assume access to $U^{\dagger}$.

\paragraph{Linear Problems.} A linear problem is characterised by a family of matrices $\mathcal{M}:=\{M_n\in\F^{n\times n}\}_{n\in\N}$, where on input $v\in\F^n$ the solution to the problem is the vector $Mv$, omitting the subscript on $M$ for readability. \\

\noindent{$\quad$\small\underline{\textsc{Linear Problem $\mathcal{M}$}}} \\
\noindent\textbf{Input:} A vector $v\in\F^n$.\\
\noindent\textbf{Output:} The matrix-vector product $Mv$.\\

This notion captures a wide variety of problems, including such fundamental ones as computing Discrete Fourier transforms, and polynomial evaluation (Vandermonde matrices). As discussed in the introduction, in this paper we show that for \emph{every} linear problem $\mathcal{M}$ there exists an efficient quantum worst-case to average-case reduction, a result for which no analogue is known in the case of uniform classical algorithms.



Next, we define the central notion that we study in this paper, namely the average-case behaviour of quantum algorithms.

\subsection{Average-case quantum algorithms}
An average-case quantum algorithm $\ALG$ is one that succeeds with probability at least $\alpha$ in expectation over (uniformly) random inputs. That is, if $\ALG$ is to compute some function $f$ mapping some known measurable domain $V$ to some co-domain $W$, then
\begin{equation}
\label{eq:avg-defn}
\Pr_{\substack{\ALG\\ v\in V}}[\ALG(v) = f(v)]  :=
        \E_{v\in V}\left[|\Pi_{f(v)}\ALG(v)|^2\right] \geq \alpha\;.
\end{equation}
Here $\Pr_{\ALG}$ denotes the probability over the internal (quantum) randomness of the algorithm arising from its unitary nature and final measurements. Note that the probability above is taken over the inputs \textit{as well as} the internal quantum randomness of the algorithm. This is highlighted by the notation we use, which we elaborate more on below. 

This notion of average-case quantum algorithms immediately suggests considering the following natural modification of quantum oracles.

\paragraph{Noisy quantum oracles.} For a function $f:\F^n\to\F^m$, we can consider a unitary oracle $U_{f}$ which on input $x\in\F^n, z\in\F^m$ performs the map
\begin{equation}
    \label{eq:bounded-error-oracle}
    U_{f}\ket{x}\ket{z}=\beta_{\mathrm{succ}}^x\ket{x}\ket{z+ f(x)} + \beta_{\mathrm{fail}}^x\ket{x}\ket{\psi(x)} \;,
\end{equation}
where $\beta_{\mathrm{succ}}^x, \beta_{\mathrm{fail}}^x\in\C$ such that $\abs{\beta_{\mathrm{succ}}^x}^2+\abs{\beta_{\mathrm{fail}}^x}^2$=1, the normalised state $\ket{\psi(x)}$ could be an arbitrary superposition over $\ket{z+v}$ for vectors $v\in\F^m\setminus\{f(x)\}$, and $+$ denotes component-wise addition for vectors over $\F$. We can interpret $U_{f}$ as a quantum analogue of a classical probabilistic algorithm for computing $f$ --- for an input $x\in\F^n$, it outputs the correct value of $f(x)\in\F^m$ with probability $\abs{\beta_{\mathrm{succ}}^x}^2$ when the second register is measured in the computational basis.

More generally, we consider oracles that may entangle a workspace register with the output
\begin{equation}
    \label{eq:bounded-error-ent-oracle}
    U_{f}\ket{x}\ket{\mathrm{w}}\ket{z}=\beta_{\mathrm{succ}}^x\ket{x}\ket{z+f(x)}\ket{\mathrm{w}(x,z,f(x))} + \beta_{\mathrm{fail}}^x\ket{x}\ket{\Psi(x)} \;,
\end{equation}
where $\ket{\Psi(x)}$ is now a normalised state of the form
\begin{equation}
    \ket{\Psi(x)}=\sum_{\substack{v\in\F^n\\v\neq f(x)}}\gamma_v^x\ket{z+v} \ket{\mathrm{w}(x,z,v)}\;.
\end{equation}

\paragraph{Quantum algorithms for linear problems.} A quantum algorithm for a linear problem $\mathcal{M}$ outputs the correct answer with some probability arising from measurement. Such an algorithm is represented by a unitary ${\ALG}$ which on an input vector $v\in\F^n$ has the action
\begin{equation}
    \label{eq:BEQLA}
    {\ALG}\ket{v}\ket{0}=\beta_{\mathrm{succ}}^v\ket{v}\ket{Mv}\ket{\mathrm{w}_0(v)} + \beta_{\mathrm{fail}}^v\ket{v}\ket{\Psi(v)},
\end{equation}
where $\beta_{\mathrm{succ}}^v,\beta_{\mathrm{fail}}^v\in\C$ are complex amplitudes such that $|\beta_{\mathrm{succ}}|^2+|\beta_{\mathrm{fail}}^v|^2=1$, $\ket{\mathrm{w}_0(v)}$ is an arbitrary normalised state of a workspace register, and 
\[
    |\beta_{\mathrm{succ}}^v|^2=\norm{\mathbbm{1}\otimes\Pi_{Mv}{\ALG}\ket{v}\ket{0}}^2
\] 
is the success probability of the algorithm, where $\Pi_{Mv}:=\ket{Mv}\!\bra{Mv}$ is an orthogonal projection on to the correct answer subspace of the output register. $\ket{\Psi(v)}$ is a normalised state that is orthogonal to the state $\ket{Mv}$ that encodes the correct matrix vector product, i.e.\ $\forall v\in\F^n, ~\ip{Mv\!~|~\!\Psi(v)}=0$. In general, it takes the form
\begin{equation}
    \ket{\Psi(v)} = \sum_{\substack{z\in\F^n\\z\neq Mv}}\gamma_z^v\ket{z}\ket{\mathrm{w}(v,z)}.
\end{equation}


%
\section{Quantum worst-case to average-case reductions}
\label{sec:main}

In this section, we provide our quantum worst-case to average-case reductions. To obtain our results in a modular way, we start by stating a general technical lemma from which we can easily derive both \cref{thm:main-circuit} and \cref{thm:main-query}.

\begin{restatable}{lemma}{lemmain}
\label{lem:main}
    Let $\F = \F_p$ be a prime field, $n \in \N$, and $\alpha\coloneqq\alpha(n) \in (0,1]$.
    Let $\ALG^M$ be a quantum query algorithm
     that has oracle access to a matrix $M\in \F^{n \times n}$, gets as input  a vector $v \in \F^{n}$, makes $Q(n)$ queries,
    and attempts to compute $Mv$. Then, for every constant $\delta>0$, there exists a quantum algorithm $(\ALG')^M$ that makes $O(\alpha^{-7/2})$ uses of $\ALG^M$, $O\left((Q(n) + n^{3/2})\cdot\alpha^{-7/2}\right)$ queries to $U_M$ and $U_M^{\dagger}$, uses $O(n^{3/2}\alpha^{-7/2}\cdot\log{n}\cdot \poly(\log\abs{\F}))$ additional one-qubit and two-qubit gates, and $O(\alpha^{-2} \cdot n\log n)$ ancillary qubits such that the following holds.
    
    For every matrix $M\in\F^{n\times n}$ such that $\ALG^M$ computes $Mv$ correctly with probability $\alpha$:
     \begin{equation*}
        \Pr_{v,\ALG}[\ALG^M(v) = Mv]  =
        \E_{v\in\F^n}\left[\norm{\Pi_{Mv}{\ALG^{M}}\ket{v}\ket{0}}^2\right] \geq \alpha\;,
    \end{equation*}
    the algorithm $(\ALG')^M$ computes $Mv$ correctly for every~$v\in\F^n$ with probability $1-\delta$:
    \begin{equation*}
    \forall v \in \F^n \quad
        \Pr_{\ALG'}[(\ALG')^M(v) = Mv]  =
        \norm{\Pi_{Mv}{(\ALG')^M}\ket{v}\ket{0}}^2 \geq 1 - \delta\;.
    \end{equation*}
\end{restatable}

Here and in the following, when we say that $\ALG'$ makes $k$ calls to $\ALG$, we mean that $\ALG'$ on inputs of length~$n$ makes $k$ calls to $\ALG$ on inputs of length~$n$.

We prove \cref{lem:main} in \cref{sec:reduction}, after we develop a toolkit of quantum algorithms in \cref{sec:quantumtools} and a quantum local correction lemma in \cref{sec:AC}, which will be necessary for our proof.

Equipped with the technical lemma above, in \cref{sec:main-all} we formally restate \cref{thm:circuit-informal} and show how it follows from \cref{lem:main}. Then, in \cref{sec:main-mvm}, we formally restate \cref{thm:query-informal} and show how to prove it using \cref{lem:main}.

\subsection{Quantum algorithms for all linear problems}
\label{sec:main-all}

In this section we deal with uniform quantum algorithms; that is, uniformly generated families of quantum circuits $\{C_n\}_{n \in \N}$ for each input length $n \in \N$. Following the literature on quantum algorithms for matrix problems
we allow $C_n$ to use a unitary gate $U_M$ that computes matrix entries of $M$ as defined in \cref{eq:quant-input-oracle-mat}. We obtain the following quantum worst-case to average-case reduction for linear problems.

\begin{theorem}
\label{thm:main-circuit}
    Let $\F = \F_p$ be a prime field, $n \in \N$, and $\alpha\coloneqq\alpha(n) \in (0,1]$. Let $\mathcal{M}:=\{M_n\in\F^{n\times n}\}_{n\in\N}$ be any linear problem, and $\ALG$ be an average-case quantum algorithm for $\mathcal{M}$ that takes as input a vector $v \in \F^{n}$ and satisfies
    \begin{equation*}
        \Pr_{v,\ALG}[\ALG(v) = Mv]  =
        \E_{v\in\F^n}\left[\norm{\Pi_{Mv}{\ALG}\ket{v}\ket{0}}^2\right] \geq \alpha\;.
    \end{equation*}

    Then, for every constant $\delta>0$, there exists a worst-case quantum algorithm $\ALG'$ satisfying
    \begin{equation*}
    \forall v \in \F^n \quad
        \Pr_{\ALG'}[\ALG'(v) = Mv]  =
        \norm{\Pi_{Mv}{\ALG'}\ket{v}\ket{0}}^2 \geq 1 - \delta\;.
    \end{equation*}
    $\ALG'$ makes $O(\alpha^{-7/2})$ uses of $\ALG$, $O(n^{3/2}\alpha^{-7/2})$ uses of $U_M$ and $U_M^{\dagger}$, $O(n^{3/2}\alpha^{-7/2}\log{n}\cdot \poly(\log\abs{\F}))$ additional one-qubit and two-qubit gates, and $O(\alpha^{-2} \cdot n\log n)$ ancillary qubits.
\end{theorem}
This result follows immediately from \cref{lem:main}. Observe that \cref{thm:main-circuit} instantiates a quantum query algorithm $\ALG$ as a family $\{C_n\}_{n \in \N}$ of quantum circuits, and considers them as having access to explicit circuits implementing $U_M$. Thus in particular, the guarantees of \cref{lem:main} about the number of queries to $U_M$ and the number of additional gates used by $\ALG'$ hold. There are no additional overheads in circuit size from implementing access to sampled vectors, which only cost $\widetilde{O}(n)$ gates.

In general, an algorithm that solves a linear problem may exploit the special structure of $M$, or in the circuit model the matrix entries may be hard-coded. However, note that there is always a trivial circuit of size $O(n^2)$ that solves a linear problem, both in the classical and quantum settings. Our focus in this work is on fine-grained complexity, and our results hold for the stronger case of non-trivial quantum circuits $\{C_n\}_{n \in \N}$ of \textit{sub-quadratic} size. 

\subsection{Matrix-vector multiplication in the query model}
\label{sec:main-mvm}

Next, we provide a quantum worst-case to average-case reduction for the Matrix-Vector Multiplication problem in the quantum query model. 

\begin{theorem}[Query complexity of Matrix-Vector Multiplication]
\label{thm:main-query}
Let $\F = \F_p$ be a prime field, $n \in \N$, and $\alpha\coloneqq\alpha(n) \in (0,1]$.
    Suppose that there exists a quantum query algorithm
    $\ALG$ that has oracle access to a matrix $M\in \F^{n \times n}$ and a vector~$v$, makes $Q(n)$ queries, and satisfies
    \begin{equation*}
        \Pr_{\substack{M,v,\\ \ALG}}[\ALG^{M,v} = Mv]  =
        \E_{\substack{M \in \F^{n\times n}\\v\in\F^n}}\left[\norm{\Pi_{Mv}{\ALG}^{M,v}\ket{0}}^2\right] \geq \alpha\;.
    \end{equation*}
    Then, for every constant $\delta>0$, there exists a worst-case  quantum algorithm $\ALG'$ that makes $O(\alpha^{-9/2})$  uses of $\ALG$, makes $O\left((Q(n) + n^{3/2})\cdot\alpha^{-9/2}\right)$ queries to $U_M$ and $U_M^{\dagger}$, and succeeds on all inputs with high probability:
    \begin{equation*}
    \forall M \in \F^{n\times n}, v \in \F^n \quad
        \Pr_{\ALG'}[(\ALG')^{M,v} = Mv]  =
        \norm{\Pi_{Mv}{(\ALG')}^{M,v}\ket{0}}^2 \geq 1 - \delta\;.
    \end{equation*}
\end{theorem}

We will prove this worst-case to average-case reduction for the Matrix-Vector Multiplication problem in the quantum query model using \cref{lem:main} and an efficient quantum verification algorithm for matrix-vector products (see \cref{lem:inf-verification} from \cref{tech:qalgo}, or \cref{lem:qVerify-simple} for a formal statement).
\begin{proof}[Proof of \cref{thm:main-query}]
The algorithm $(\ALG')^{M,v}$ repeats the following procedure $O(1/\alpha)$ times. Sample a uniformly random matrix $R \in \F^{n \times n}$, define $M'=M-R$, and generate unitary oracle $U_{M'}$ for $M'$. Each query to $U_{M'}$ costs only one query to $U_M$. Now use \cref{lem:main} with $U_{M'}$ to try to compute
     $b'= M' \cdot v$ . Finally, compute $b_R = R \cdot v$ directly without any queries to~$M$, and set $b=b' + b_R$. Verify whether $b=Mv$ using \cref{lem:qVerify-simple} with $\eps=\delta/2$, and output $b$ if the verification test passed. 
     
The complexity of the algorithm is determined by $O(1/\alpha)$ applications of \cref{lem:main} and verifications from \cref{lem:qVerify-simple}, and $O(n)$ queries to read the coordinates of~$v$. In particular, the described algorithm $(\ALG')^{M,v}$ performs $O(\alpha^{-9/2})$ calls of $\ALG^{M,v}$, and $O\left((Q(n) + n^{3/2})\cdot\alpha^{-7/2}\right)$ queries to $U_M$ and $U_M^{\dagger}$.

In order to analyze the correctness of the algorithm $(\ALG')^{M,v}$, we introduce the following notation. For a fixed matrix $M\in\F^{n\times n}$, let $p_M=\Pr_{v,\ALG}[\ALG^{M,v}=Mv]$ denote the probability of computing $\ALG^{M,v}=Mv$ correctly for the fixed value of~$M$ and a uniformly random vector $v \in \F^n$. Let us define the set $X\subseteq\F^{n\times n}$ of good matrices where the average-case algorithm $\ALG$ succeeds with probability at least $\alpha/2$:
\begin{align*}
            X\colon= \left\{M \in \F^{n\times n} \colon p_M \geq \frac{\alpha}{2}\right\} \; .
    \end{align*}
    
First we observe that $|X|\geq \alpha |\F|^{n^2}/2$. Indeed, by the assumption of the theorem we have that $\E_{M \in\F^{n \times n}}[p_M] \geq \alpha$. Then,
    \begin{eqnarray*}
        \alpha \leq \E_M[p_M]
        < 1 \cdot \Pr_M[p_M \geq \alpha/2] + \alpha/2 \cdot \Pr_M[p_M < \alpha/2] 
        \leq \Pr_M[p_M \geq \alpha/2] +
        (\alpha/2) \cdot 1 \;.
    \end{eqnarray*}
    Thus, $|X|=|\F|^{n^2}\cdot\Pr_M[p_M \geq \alpha/2] \geq \alpha |\F|^{n^2}/2$.

Since in every iteration of the algorithm, the matrix $M'=M-R$ is a uniformly random matrix, we have that $\Pr[M'\in X] \geq \alpha/2$. In the case when $M'\in X$, \cref{lem:main} can correctly compute $b'=M'v$ with probability $2/3$. Thus, in every iteration, the described procedure computes $b=Mv$ with probability $\Omega(\alpha)$. By repeating this procedure $O(1/\alpha)$ times (each time verifying whether $b=Mv$ using \cref{lem:qVerify-simple} with $\eps=\delta/2$), we have that for every matrix~$M$ and every vector~$v$, we compute the product $Mv$ correctly with probability $1-\delta$.
\end{proof}
\section{Quantum toolkit for local correction}
\label{sec:quantumtools}
In this section, we provide a toolkit of quantum algorithms that will allow us to later obtain a quantum local correction lemma (in \cref{sec:AC}), which will underlie our quantum worst-case to average-case reductions. Specifically, in \cref{sec:quantumtools-verification}, we construct a quantum procedure for flagging correct matrix-vector products in superposition; we use this in \cref{sec:indicator} to obtain a unitary implementation of an approximation to the indicator function on the set $X$ of good vectors, and in \cref{sec:quantum-sampling} to provide a quantum procedure for efficiently sampling from $X$; finally, in \cref{sec:quantumtools-fourier}, we construct a quantum procedure for learning Bogolyubov subspaces from noisy quantum oracles.

\subsection{Flagging correct matrix-vector products in superposition}
\label{sec:quantumtools-verification}

As discussed in the technical overview in \cref{sec:overview}, a key bottleneck in the classical setting is the efficient verification of computing a matrix-vector product. In the setting of quantum query complexity, this is the following problem:\\

\noindent{$\quad$\small\textsc{\underline{Matrix-vector Product Verification ($\mathsf{MvPV}$)}}}\\
\noindent\textbf{Input:} Quantum oracles $U_M$, $U_v$, and $U_b$ for a matrix $M\in\F^{n\times n}$ and vectors $v,b\in\F^n$.\\
\noindent\textbf{Output:} $1$ if $Mv=b$ and $0$ otherwise.\\

Classical query algorithms would analogously have access to an oracle that returns matrix (vector) entries when queried with a row and column index pair.

The quantum query complexity of $\mathsf{MvPV}$ has been studied in detail, starting with the work of \cite{BuhrmanSpalek04} who showed that matrix-matrix products can be verified with $O(n^{5/3})$ quantum queries, a bound that is sublinear in the size of the input (which is $n^2$). \cite{Magniez11WalkSearch} later showed that the classical techniques of Freivalds can be adapted to the quantum setting to make this algorithm time-efficient. The special case of matrix multiplication over the Boolean semiring has close relations to path and triangle finding in graphs, and its quantum query complexity has also been studied in great depth \cite{Magniez07Triangle,Vassilevka10MatrixProduct,LeGall12BooleanMatrix,Childs12BooleanMatrix,Jeffery13BooleanMatrix}. \cite{KothariThesis}, gives a detailed review of the complexities and relationships between different variants of the matrix multiplication and $\mathsf{MvPV}$ problems over arbitrary semirings. 

\begin{table}[htb]
    \renewcommand{\arraystretch}{1.5}
    \centering
    \begin{adjustbox}{center}
    \begin{tabular}{|l|L|L|}
        \hline
        \textbf{Problem} & \textbf{Inputs/Output} & \textbf{Quantum query complexity} \\ 
         \hline
         Matrix-vector Multiplication ($\mathsf{MvM}$)  & In: $M\in\F^{n\times n},v\in\F^n$ Out: $Mv$ & $\Theta(n^2)$\\
         \hline
         Matrix-vector Product Verification ($\mathsf{MvPV}$) & In: $M\in\F^{n\times n},v,b\in\F^n$ Out: $Mv=b?$ & $\Theta(n^{3/2})$ \\
         \hline
    \end{tabular}
    \end{adjustbox}
    \caption{Worst-case quantum query complexity of $\mathsf{MvM}$ and $\mathsf{MvPV}$ \cite{KothariThesis}. }
    \label{tab:MvPV}
\end{table}
%

This extensive literature focuses on the more general problem of matrix-matrix multiplication and product verification over semirings, especially the Boolean case; on the other hand, our interest here lies in the special case of matrix-vector products over finite fields $\F_p$. In particular, all the results about matrix-vector product verification that we have come across in past work deal exclusively with the query complexity, leaving the algorithm achieving the upper bound, and its computational or gate complexity, implicit. Furthermore, techniques such as those used in \cite{BuhrmanSpalek04} do not extend well to  our setting --- (1) because they use the computation of a logical $\mathsf{AND}$ via Grover search, they do not directly lead to a unitary algorithm that can be queried in superposition, due to the traditional problems of overshooting associated with vanilla quantum search; and (2) because they are phrased in the usual manner of performing amplitude amplification followed by measurements to extract the output with constant success probability.

In this section, we construct an efficient quantum algorithm for addressing $\mathsf{MvPV}$, with query and gate complexities bounded by the optimal $O(n^{3/2})$, which will be conducive to querying in superposition and consequently to composition with our subsequent subroutine for learning Bogolyubov subspaces. As discussed in \cref{tech:qalgo}, we the main technical ingredients we use are fixed-point amplitude amplification and quantum singular value threshold projection.

\subsubsection{A simple case of $\mathsf{MvPV}$}
We first consider the standard setting where the input vectors $v$ and $b$ are both given by the usual exact quantum query oracles. We later use this simple variant to argue about the case where we only have a noisy version of $b$, accessed via a noisy quantum oracle.

\begin{lemma}[Quantum $\mathsf{MvPV}$]
\label{lem:qVerify-simple}
    Suppose we are given a quantum oracle $U_M$ for a matrix $M\in\F^{n\times n}$ 
    \begin{equation}
        U_M\ket{j, k, z} = \ket{j, k, z \oplus M_{jk}}
    \end{equation}
    for all indices $j,k\in [n]$ and $z\in\F$, and quantum oracles $U_v, U_b$ for input vectors $v,b\in\F$ 
    \begin{align}
        U_v\ket{j, z} &= \ket{j,  z \oplus v_{j}}\\
        U_b\ket{j, z} &= \ket{j,  z \oplus b_{j}}
    \end{align}
    for all $j\in[n]$ and $z\in\F$. Then there is a gate-efficient quantum algorithm $\mathsf{Q_{verify}}$ that accepts with certainty if $Mv=b$, and rejects with probability $1-\epsilon$ if $Mv\neq b$. Furthermore, $\mathsf{Q_{verify}}$ uses $q=O\left(n^{3/2}\cdot \log\frac{1}{\epsilon}\right)$ queries to $U_M$ and $U_M^{\dagger}$, $O(n)$ queries to $U_v$, $U_b$ and their Hermitian conjugates, $O\left( q\log n\cdot\poly\log |\F|\right)$ additional one-qubit and two-qubit gates, and $O(n\log |\F|)$ ancillary qubits. 
\end{lemma}
\begin{proof}
    We denote multiplication and controlled addition over a finite field by a $\mathsf{CCX}$ gate defined by
    \begin{equation}
        \mathsf{CCX}\ket{s_1}\ket{s_2}\ket{z} = \ket{s_1}\ket{s_2}\ket{z+s_1\cdot s_2},
    \end{equation}
    where $s_1,s_2,z\in\F$ and $+$ and $\cdot$ are addition and multiplication in $\F$. Implementing such an operation requires only $O(\poly\log |\F|)$ elementary two-qubit gates \cite{Beauregard03FiniteFieldsArithmetic}.
    Throughout this paper we assume that all algorithms use quantum registers of dimension $p = |\F|$ and perform arithmetic over $\F$, since such arithmetic can be simulated using $O(\log |\F|)$ qubits with $O(\poly\log |\F|)$ overhead in gate complexity. 
    
    \vspace{1mm}    
        \begin{tcolorbox}[title=Circuit~{$U_{Mv}$}: Quantum circuit for entrywise Matrix-vector product,
            standard jigsaw,
            opacityback=0] 
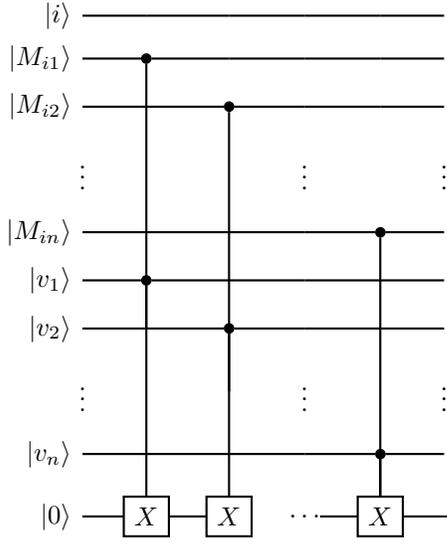
        \begin{multicols}{2}
            \small
            \centering
%
\begin{quantikz}
    \lstick{$\ket{i}$} & \qw
    & \qw & \qw & \qw & \qw \\
    \lstick{$\ket{M_{i1}}$} & \ctrl{5}
    & \qw & \qw & \qw & \qw \\
    \lstick{$\ket{M_{i2}}$} & \qw & \ctrl{5}
    & \qw & \qw & \qw \\
    \vdots &&& \vdots && \vdots \\ 
    \lstick{$\ket{M_{in}}$} & \qw & \qw & \qw & \ctrl{5} & \qw \\
    \lstick{$\ket{v_1}$} & \ctrl{4} & \qw & \qw & \qw & \qw \\
    \lstick{$\ket{v_2}$} & \qw & \ctrl{3} & \qw & \qw & \qw \\
    \vdots &&& \vdots && \vdots \\ 
    \lstick{$\ket{v_n}$} & \qw & \qw & \qw & \ctrl{1} & \qw \\
    \lstick{$\ket{0}$} & \gate{X} & \gate{X} & \ldots & \gate{X} & \qw & \qw\rstick{$\ket{(Mv)_i := \sum_{j=1}^n M_{ij}v_j}$}
\end{quantikz}

%
            \columnbreak
            \paragraph{$U_{Mv}$ ---} Quantum circuit that uses $n$ oracle calls each to $U_M$, $U_v$, $U_M^{\dagger}$, $U_v^{\dagger}$, and $n\cdot O(\poly\log |\F|)$ elementary gates, to implement oracle access to the entries of the Matrix-vector product $Mv$. The index $i$ is given in an additional register, using which the matrix entries $M_{ij}$ and vector entries $V_j$ are loaded by oracle calls to $U_M$ and $U_v$ (suppressed here for readability); accounting for all ancillary registers, the width of the circuit is $O(n\log |\F|)$ qubits, and its depth is $O(n\,\poly\log |\F|)$. 
        \end{multicols}
        \end{tcolorbox}
    \vspace{2mm}    
The unitaries $U_{Mv}$, $U_b$ and their Hermitian conjugates can be used once each, along with quantum arithmetic operations over $\F$ using circuits of size $O(\poly\log |\F|)$, to obtain a similar oracle $U_{Mv-b}$ for the vector $Mv-b$. The registers used for computing $(Mv)_i$ and $b_i$ can both be perfectly uncomputed and returned to their initial value in this case.
Using quantum adder circuits \cite{Gidney2018halvingcostof} it is possible to construct a comparator circuit with $O(1)$ ancillas and $O(\log |\F|)$ gates 
that checks whether an entry $(Mv-b)_i$ is zero, and sets a flag qubit to $1$ if not. This becomes an oracle to the component-wise indicator function $\mathsf{1}_{(Mv-b)}:[n]\times\F\to\{0,1\}$ mapping an index $i\in [n]$ to $0$ if $(Mv)_i=b_i$, and $1$ otherwise. Denote this oracle by $U_{Mv\stackrel{?}{=}b}$.

\vspace{1mm}    
        \begin{tcolorbox}[title=Circuit~{$U_{Mv-b}$}: Quantum circuit for $Mv-b$,
            standard jigsaw,
            opacityback=0] 
            \small
            
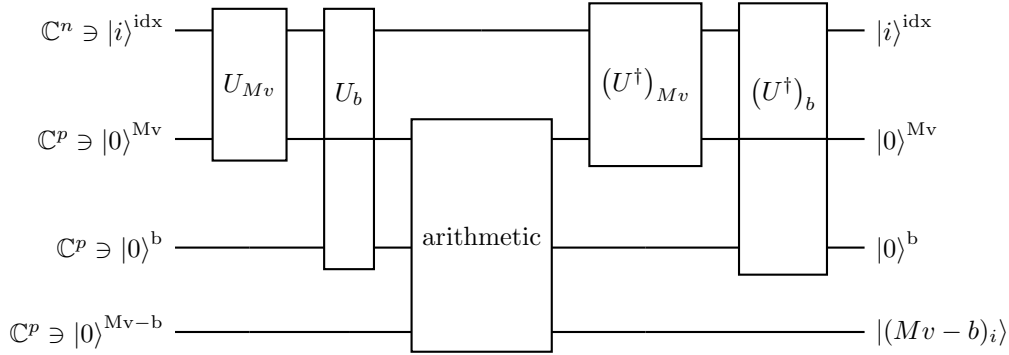
\centering
                \begin{quantikz}[transparent]
    \lstick{$\C^{n}\ni\ket{i}^{\mathrm{idx}}$} & \gate[2]{U_{Mv}} & \gate[3,label style={yshift=0.6cm}]{U_{b}} & \qw & \gate[2]{\left(U^{\dagger}\right)_{Mv}} & \gate[3,label style={yshift=0.6cm}]{\left(U^{\dagger}\right)_{b}} & \qw\rstick{$\ket{i}^{\mathrm{idx}}$} \\
    \lstick{$\C^{p}\ni\ket{0}^{\mathrm{Mv}}$} & & \linethrough & \gate[3]{\text{arithmetic}} & & \linethrough &\qw\rstick{$\ket{0}^{\mathrm{Mv}}$}  \\
    \lstick{$\C^{p}\ni\ket{0}^{\mathrm{b}}$} & \qw & & & \qw & & \qw\rstick{$\ket{0}^{\mathrm{b}}$} \\
    \lstick{$\C^{p}\ni\ket{0}^{\mathrm{Mv-b}}$} & \qw & \qw & & \qw & \qw & \qw \rstick{$\ket{(Mv-b)_i}$}\\
\end{quantikz}

            \paragraph{$U_{Mv-b}$ ---} Quantum circuit to implement oracle access to the entries of $Mv-b$ (where we suppress the workspace registers required by $U_{Mv}$ etc.). The width of the circuit is $O(n\log |\F|)$ qubits, and the arithmetics step can be implemented with circuits of size $O(\poly\log |\F|)$.
        \end{tcolorbox}
    \vspace{2mm}    

We can then perform a quantum search for nonzero entries in $Mv-b$ using $U_{Mv\stackrel{?}{=}b}$. We start by preparing the uniform superposition over indices, query $U_{Mv\stackrel{?}{=}b}$, and treat nonzero entries as marked. 

Define the states $\ket{\psi_0}$, $\ket{\psi_1(v,b)}$, and $\ket{\psi_0(v,b)}$ as the uniform superposition over all indices, and uniform superpositions over indices at which $Mv-b$ has nonzero and zero entries respectively
\begin{align}
\label{eq:psi0-psi1-defn}
    \ket{\psi_0}&=\frac{1}{\sqrt{n}}\sum_{i\in[n]}\ket{i} \nonumber\\
    \ket{\psi_1(v,b)} &= \frac{1}{\sqrt{m_{vb}}}\sum_{\substack{i\in[n]\\(Mv)_i\neq b_i}}\ket{i} \nonumber\\
    \ket{\psi_0(v,b)} &= \frac{1}{\sqrt{n-m_{vb}}}\sum_{\substack{i\in[n]\\(Mv)_i=b_i}}\ket{i},
\end{align}
and $m_{vb}=|\{i\in[n]~:~(Mv)_i\neq b_i\}|$. Let $k:=\lceil\log n\rceil$. Without loss of generality we can assume $n$ is a power of two since we can pad vectors with zeros if not, while at most doubling the complexity. Consider the operator that creates the uniform superposition over indices and queries $Mv\stackrel{?}{=}b$, i.e.\
\[
    U = \left(H^{\otimes k}\otimes\mathbbm{1}\right)U_{Mv\stackrel{?}{=}b}\;.
\]
We omit the workspace registers used by $U_{Mv\stackrel{?}{=}b}$ because in this case the workspace can be uncomputed and returned to the initial all-zeros state. 

The action of $U$ on the $\ket{0^{k+1}}$ state is
\begin{align}
\label{eq:fpaa-unitary}
    U\ket{0^{k+1}} = \begin{cases}
        \qquad\ket{\psi_0}\ket{0} & Mv=b\\
        \quad\\
        \sqrt{\frac{n-m_{vb}^{}}{n}}\ket{\psi_0(v,b)}\ket0 + \sqrt{\frac{m_{vb}^{}}{n}}\ket{\psi_1(v,b)}\ket1 & Mv\neq b. 
    \end{cases}
\end{align}

Since we do not know $m_{vb}$ beforehand, and since we will need a unitary subroutine that is run over a superposition of input vectors in $\F^n$ later on, we use fixed point amplitude amplification \cite{Grover2005FixedPointSearch,Yoder14fpoaa,Guerreschi19fpoaa,Gilyen2019}. This allows the number of iterations of amplitude amplification to be chosen uniformly for all vectors $v$ without worrying about the problem of under- or over-shooting faced in vanilla amplitude amplification. For completeness, we give a brief overview and technical statement of fixed-point amplitude amplification in \cref{app:qsvt}.


The unitary $U$ of \cref{eq:fpaa-unitary} satisfies the conditions required by \cref{thm:fpaa} with $\Pi=\mathbbm{1}_k\otimes\ketbra{1}{1}$, preparing a state that has a component flagged with $\ket1$ when $Mv\neq b$. In the worst case for $Mv\neq b$, there exists exactly one coordinate $i\in[n]$ at which $(Mv)_i\neq b_i$, so that $m_{vb}=1$. We hence use fixed point amplification with the worst-case lower bound of $p>\frac{1}{2\sqrt{n}}$ to obtain a unitary $U'$ with the action
\begin{align}
    U'\ket{0^{k+1}} = \begin{cases}
        \qquad\ket{\psi_0}\ket{0} & Mv=b\\
        \quad\\
        \gamma^{vb}_{\mathrm{fail}}\ket{\psi_0(v,b)}\ket0 + \gamma^{vb}_{\mathrm{succ}}\ket{\psi_1(v,b)}\ket1 & Mv\neq b,
    \end{cases}
\end{align}
where $\abs{\gamma^{vb}_{\mathrm{succ}}}^2\geq 1-\epsilon$. This unitary $U'$ makes $q=O\left(\sqrt{n}\cdot\log\frac{1}{\epsilon}\right)$ queries to $U$ and $U^{\dagger}$, uses $O(q\log n)$ additional elementary gates, and a single ancilla. The only queries in $U$ are those made by $U_{Mv\stackrel{?}{=}b}$ to $U_M$, $U_v$, $U_b$ and their conjugates; since $U_{Mv\stackrel{?}{=}b}$ uses $n$ queries to each of these oracles, the net query complexity of $U'$ is $q'=O(nq)=O(n^{3/2})$. Finally since $M$ is of size $n^2$ and $q'=O(n^{3/2})$ we can always first query $v$ and $b$ into an ancillary register, since both are vectors of length $n$.

Effectively, for any desired success probability $\epsilon\in(0,1)$ we are able to obtain a unitary $\mathsf{Q_{verify}}:=U'$ that computes $\mathsf{1}_{Mv=b}$ with one-sided bounded error, measuring the output register of which yields the claimed guarantees.
\end{proof}

\subsubsection{Flagging the correct matrix-vector product in superposition}
In order to establish our reduction, we need to go beyond verifying whether a certain vector, obtained by running $\ALG$ on an input $v\in\F^n$ and measuring its output register, is correct. In particular, we would like to compute the indicator function $\mathsf{1}_{X}$ which marks the vectors on which $\ALG$ succeeds with high probability in superposition over all $v\in\F^n$, in order to learn its Fourier characters of high weight. This is complicated twofold --- by the noise or part of the state associated with failure in the output of $\ALG$, and the entanglement with workspace registers. Nevertheless, we can compute an approximate, noisy version of the indicator function on good input vectors. As a first step, we construct an algorithm using ideas from \cref{lem:qVerify-simple} to ensure that whenever $\ALG$ outputs the correct answer, it also outputs a flag indicating success.

\begin{lemma}[Noisy quantum $\mathsf{MvPV}$]
\label{lem:qVerification}
    Suppose we are given a quantum oracle $U_M$ for a matrix $M\in\F^{n\times n}$ 
    \begin{equation}
        U_M\ket{j, k, z} = \ket{j, k, z \oplus M_{jk}}
    \end{equation}
    for all indices $j,k\in [n]$ and $z\in\F$, 
    and a noisy quantum algorithm $\ALG$ as described in \cref{eq:BEQLA}, i.e.\
    \[
        \ALG\ket{v}\ket{0}=\beta_{\mathrm{succ}}^v\ket{v}\ket{Mv}\ket{\mathrm{w}_0(v)} + \beta_{\mathrm{fail}}^v\ket{v}\ket{\Psi(v)}.
    \]
    Then there exists a gate-efficient quantum algorithm $\ALG_{\mathsf{verified}}$ that succeeds and outputs $Mv$ with probability $\abs{\beta_{\mathrm{succ}}^v}^2$ along with a flag indicating success, and similarly outputs a flag indicating failure whenever it outputs a vector $z\neq Mv$, with probability at least $(1-\epsilon)\abs{\beta_{\mathrm{fail}}^v}^2$. Furthermore, $\ALG_{\mathsf{verified}}$ uses $q=O\left(n^{3/2}\cdot\log\frac{1}{\epsilon}\right)$ queries to $U_M$ and $U_M^{\dagger}$, $O(1)$ queries to $\ALG$, $O\left(q\log n\cdot\poly\log |\F|\right)$ additional one-qubit and two-qubit gates, and $O(n\log |\F|)$ ancillary qubits. 
\end{lemma}
\begin{remark}
    The input vectors $v$ and $b$ are taken to be given as states $\ket{v}$ and the right hand side of \cref{eq:BEQLA}---created using a single query to $\ALG$--- respectively. Using a hard-coded circuit of size $O(n\log n)$ we can implement entrywise oracles $U_v$ and $U_b$ that copy out the $i^{th}$ entries of these vectors into a target register controlled on an index register.
\end{remark}

\begin{proof}
If $\ALG$ had the property that either $\beta_{\mathrm{fail}}^v=0$ or $\beta_{\mathrm{succ}}^v=0$ always, then we could have proceeded as in \cref{lem:qVerify-simple}. Nevertheless, if we first prepare the state $\ALG\ket{v}\ket0$, prepare a superposition over indices $i\in[n]$ in an ancillary register, and compute the indicator function $\mathsf{1}_{(Mv-b)}$ using fixed circuits of size $O(n\log n)$ gates to copy (controlled on an index register) the $i^{th}$ entry from the output registers of $\ALG$, we now have a unitary $U$ that prepares the state
\begin{align}
\label{eq:computing-Av-w}
    &U\biggl(\ket{v}\otimes\left(\beta_{\mathrm{succ}}^v\ket{Mv}\ket{\mathrm{w}_0(v)} + \beta_{\mathrm{fail}}^v\left(\sum_{\substack{z\in\F^n\\z\neq Mv}}\gamma_z^v\ket{z}\ket{\mathrm{w}(v,z)}\right)\right)\otimes\ket{0^{k+1}}\biggr)\nonumber\\
        &=\ket{v}\otimes\left(\vphantom{\sum_{\substack{z\in\F^n\\z\neq Mv}}}\beta_{\mathrm{succ}}^v\ket{Mv}\ket{\mathrm{w}_0(v)}\ket{\psi_0}\ket{0} +\right.\nonumber\\
        &\qquad\qquad\qquad \left. \beta_{\mathrm{fail}}^v\left(\sum_{\substack{z\in\F^n\\z\neq Mv}}\gamma_z^v\ket{z}\ket{\mathrm{w}(v,z)}\biggl[\sqrt{\frac{n-m_{vz}}{n}}\ket{\psi_0(v,z)}\ket0+\sqrt{\frac{m_{vz}}{n}}\ket{\psi_1(v,z)}\ket1\biggr]\vphantom{\beta_{\mathrm{succ}}^v\ket{Mv}}\right)\right),
\end{align}
where the states $\ket{\psi_0},\ket{\psi_0(v,z)},\ket{\psi_1(v,z)}$ are defined analogously to \cref{eq:psi0-psi1-defn}, $m_{vz}$ is the number of indices at which a vector $z$ disagrees with the correct  answer $Mv$, and to avoid clutter we have not explicitly written the registers used for querying the entries $M_{ij}$, since they can be uncomputed as before.

As in \cref{lem:qVerify-simple}, we can apply fixed-point amplitude amplification to the unitary $U$, with the goal of amplifying the part of the superposition flagged by $1$ in the last qubit. Incorrect answers $z\neq Mv$ may in general be wrong only at a single co-ordinate, so that in the worst case there may be a single $\gamma_z^v=1$ with the corresponding $m_{vz}=1$. Thus we can use the worst-case lower bound of $\frac{1}{\sqrt{n}}$ on the amplitude of the target state flagged by $1$. 
Since $U$ does not affect the state of the input, output, or workspace registers of $\ALG$, the amplified operator $\ALG_{\mathsf{verified}}:=U'$ will also preserve the same superposition that is produced by $\ALG$, with the component of the state flagged by $1$ in the final register amplified: 
\begin{align}
\label{eq:fpoaa-in-suppsn}
    &\ALG_{\mathsf{verified}}\biggl(\ket{v}\otimes\left(\beta_{\mathrm{succ}}^v\ket{Mv}\ket{\mathrm{w}_0(v)} + \beta_{\mathrm{fail}}^v\left(\sum_{\substack{z\in\F^n\\z\neq Mv}}\gamma_z^v\ket{z}\ket{\mathrm{w}(v,z)}\right)\right)\otimes\ket{0^{k+1}}\biggr)\nonumber\\
        &=\ket{v}\otimes\left(\vphantom{\sum_{\substack{z\in\F^n\\z\neq Mv}}}\beta_{\mathrm{succ}}^v\ket{Mv}\ket{\mathrm{w}_0(v)}\ket{\psi_0}\ket{0} +\right.\nonumber\\
        &\qquad\qquad\qquad \left. \beta_{\mathrm{fail}}^v\left(\sum_{\substack{z\in\F^n\\z\neq Mv}}\gamma_z^v\ket{z}\ket{\mathrm{w}(v,z)}\biggl[\gamma^{vz}_{\mathrm{fail}}\ket{\psi_0(v,z)}\ket0+\gamma^{vz}_{\mathrm{succ}}\ket{\psi_1(v,z)}\ket1\biggr]\vphantom{\beta_{\mathrm{succ}}^v\ket{Mv}}\right)\right),
\end{align}
where $\abs{\gamma^{vz}_{\mathrm{succ}}}^2>1-\epsilon$, and $\ALG_{\mathsf{verified}}$ uses $q=O\left(n^{3/2}\cdot\log\frac{1}{\epsilon}\right)$ queries to $U_M$ and $U_M^{\dagger}$, $O(1)$ queries to $\ALG$, and $O(q\log n\cdot\poly\log |\F|)$ additional elementary gates, and a single ancillary qubit.

Measuring the register containing the output of $\ALG$, we see that $\ALG_{\mathsf{verified}}$ outputs $Mv$ with probability at least $\abs{\beta^v_{\mathrm{succ}}}^2$ just as $\ALG$ does, but now the last register contains a flag qubit set to zero to indicate success. When $\ALG_{\mathsf{verified}}$ outputs an incorrect vector $z\neq Mv$, with probability at least $(1-\epsilon)\abs{\beta^v_{\mathrm{succ}}}^2$ the flag qubit is set to $1$ to indicate failure, giving the guarantees claimed in the lemma.
\end{proof}

\subsection{Noisy quantum oracles approximating the indicator function \texorpdfstring{$\mathsf{1}_X$}{1X}} 
\label{sec:indicator}

To obtain our next two quantum procedures: an efficient sampler for the set $X$ of good inputs to $\ALG$ and a learner for the Bogolyubov subspace, we first need to obtain a unitary implementation of an approximate version of the indicator $1_X$, which is a Boolean valued function defined on $\F^n$. We start by observing that the subroutine $\ALG_{\mathsf{verified}}$ that we constructed in the previous section almost has the basic property that a unitary implementing $\mathsf{1}_X$ should have: it attaches flags zero and one to vectors in $X$ and its complement respectively. However, it is \textit{noisy} in two ways: it errs with high probability on vectors outside $X$, and only succeeds with low probability on vectors in $X$. What we would like is (a noisy version of) the oracle
\begin{equation*}
    U_X\ket{v}\ket{0} = \ket v \ket{\mathsf{1}_X(v)}.
\end{equation*}
To obtain such a unitary, we use the machinery of quantum singular value threshold projection on top of $\ALG_{\mathsf{verified}}$. We now go into the details of our construction below, and give a brief technical statement of the quantum singular value threshold projection technique in \cref{app:qsvt}. 

For an average-case algorithm $\ALG$ with average success probability $\alpha$ as defined in \cref{eq:avg-defn}, recall that we are interested in the associated set of ``good'' input vectors defined by 
\[
    X = \{v\in\F^n : \abs{\beta_{\mathrm{succ}}}^2\geq \frac{\alpha}{2}\}\;,
\]
which we know has density at least $\alpha/2$ in $\F^n$, by the averaging argument in \cref{clm:density-of-X}. The indicator function of this set takes the value $\mathsf{1}_X(v)=1$ when $v\in X$ and $\mathsf{1}_X(v)=0$ otherwise. Denote by $U$ the algorithm $\ALG_{\mathsf{verified}}$ constructed in \cref{lem:qVerification}. Bundling all the workspace registers together for brevity, we note that it has the following property: defining the projectors $\Pi=\mathbbm{1}^{\mathrm{v}}\otimes\ketbra{0}{0}^{\mathrm{work,flag}}$ and $\widetilde{\Pi}=\mathbbm{1}\otimes\ketbra{0}{0}^{\mathrm{flag}}$, where the identity term acts on all registers except the flag qubit, we have that
\[
    \widetilde{\Pi}U\Pi = \sum_{v\in\F^n} |\beta'^v_{\mathrm{succ}}|\ketbra{w_v}{v,0,0}\;.
\]
We interpret the right hand side above as the singular value decomposition of a matrix with right singular vectors $\ket{v,0,0}$, and left singular vectors $\ket{w_v}$ given by
\begin{align}
   \ket{w_v} = \frac{1}{|\beta'^v_{\mathrm{succ}}|}\biggl(\beta^v_{\mathrm{succ}}\ket{v,Mv,\mathrm{w}_0(v),\psi_0} + \beta^v_{\mathrm{fail}}\sum_{\substack{z\in\F^n\\z\neq Mv}}\gamma_z^v\gamma^{vz}_{\mathrm{fail}}\ket{v,z,\mathrm{w}(v,z),\psi_0(v,z)}\biggr)\otimes\ket0^{\mathrm{flag}}.
\end{align}
The singular values $|\beta'^v_{\mathrm{succ}}|$ are defined by the relation
\begin{align}
    |\beta'^v_{\mathrm{succ}}|^2 = |\beta^v_{\mathrm{succ}}|^2 + |\beta^v_{\mathrm{fail}}|^2\sum_{\substack{z\in\F^n\\z\neq Mv}}|\gamma_z^v|^2|\gamma^{vz}_{\mathrm{fail}}|^2\;.
\end{align}
Fixing a threshold parameter $t\in(0,1)$, consider the partition of $\F^n$ into the following three sets:
\begin{enumerate}
    \item a good set $X_t^g=\{v\in\F^n~:~|\beta^{v}_{\mathrm{succ}}|^2\geq  t + t^2\}$
    \item an intermediate set $W_t=\{v\in\F^n~:~|\beta^{v}_{\mathrm{succ}}|^2\in \left(t-2t^2, t+t^2\right)\}$
    \item a bad set $X_t^b=\{v\in\F^n~:~|\beta^{v}_{\mathrm{succ}}|^2\leq t-2t^2\}$.    
\end{enumerate}
Choosing $\epsilon=t^2$ in \cref{lem:qVerification} we have that when $v\in X^g_t$, $|\beta'^v_{\mathrm{succ}}|^2 \geq t+t^2$, and when $v\in X^b_t$, $|\beta'^v_{\mathrm{succ}}|^2\leq t-t^2$. 

We see that $U=\ALG_{\mathrm{verified}}$ hence satisfies the conditions required in \cref{thm:qsv-projections}, so that we can use it to select the vectors in $X^g_t$ with high probability using the technique of singular value threshold projection. Since $\sqrt{t+t^2}\geq \sqrt{t}\left(1+\frac12 t-\frac18t^2\right)$, we can choose the thresholds in \cref{thm:qsv-projections} to be $t$ as above, and $\delta=\frac{1}{2}t^{3/2}-\frac18 t^{5/2}$, and obtain a unitary $U_q$ with the action
\begin{equation}
    U_q\ket{v,0,0} = \widetilde{\beta}^v_{\mathrm{succ}}\ket{w_v}\ket0^{\mathrm{flag}} + \widetilde{\beta}^v_{\mathrm{fail}}\ket{\Psi_v}\ket1^{\mathrm{flag}},
\end{equation}
where the amplitudes satisfy the following guarantees ($\Pr(\text{flag}=0):=|\widetilde{\beta}^v_{\mathrm{succ}}|^2$ etc):
\begin{enumerate}
    \item for inputs $v\in X_t^g$, $\Pr(\text{flag}=0)\geq 1-2\epsilon$;

    \item for inputs $v\in X_t^b$, $\Pr(\text{flag}=1)\geq 1-2\epsilon$.
\end{enumerate}
Importantly, we note that for inputs $v\in W_t$, we get no useful guarantee other than the consistency condition $\Pr(\text{flag}=0)\leq 1$. The unitary $U_q$ can be implemented using $q=O(\frac1\delta\log\frac1\epsilon)=O(\frac{1}{t^{3/2}}\log\frac1\epsilon)$ queries to $U$, i.e.\ to $\ALG_{\mathrm{verified}}$, and $O(q)$ additional one-qubit and two-qubit gates.

$U_q$ represents a noisy version of the indicator function on $X^g_t$ --- applying a Pauli-$\mathsf{X}$ gate to flip the flag qubit, $U_q$ is a noisy quantum oracle for $\mathsf{1}_{X}$ as defined in \cref{eq:bounded-error-ent-oracle}, with one additional complication: there is a ``wasteland'' slice $W\subseteq\F^n$ on which $U_q$ gives no guarantee as to the value of $\mathsf{1}_{X}$. Note that this can in principle be a very large set --- by inspecting the averaging argument in \cref{clm:density-of-X} we see that it is possible to construct adversarial examples of $\ALG$ for which $W$ can have density as large as $\frac{1-\alpha}{1-\alpha/2}$ for $t=\alpha/2$. Nevertheless, we will show how to overcome this difficulty in \cref{subsec:thresholds}.

\subsection{Quantum sampling from the set of good inputs}
\label{sec:quantum-sampling}
Using the subroutine $\ALG_{\mathrm{verified}}$, it is also possible to sample from the set of vectors on which $\ALG$ succeeds with at least a desired probability. 
Suppose we are interested in the set $X_{\tau}=\{v\in\F^n : \abs{\beta^v_{\mathrm{succ}}}^2>\tau\}$ for some $\tau\in(0,1)$. We use the same ideas that we presented in the previous section --- using singular value threshold projection from \cref{thm:qsv-projections} with the choice of threshold $t=\sqrt{\tau}$ and $\delta=\eta t$ for some $\eta>0$ and so $q=O(\frac{1}{\eta\sqrt{\tau}}\log\frac{1}{\epsilon})$, we get a unitary $U_q$ that implements a noisy version of an approximation to $\mathsf{1}_X$. It is sufficient to choose constant $\eta$, e.g. $\eta=0.01$, because unlike in the case of the indicator function in the previous section, we do not require high precision with respect to the wasteland slice $W_t$: instead of sampling from $X_{\tau}$, we can easily work with $X_{\tau'}\subseteq X_{\tau}$ for $\tau'=(1+\eta)\tau$ without any difficulties.

Suppose we prepare the uniform superposition over all $v\in\F^n$ and run $U_q$.
If the density of $X_{\tau}$ is $\mu(X_{\tau})$, we see that with the projector $\widetilde{\Pi}=\mathbbm{1}\otimes\ketbra{0}{0}^{\mathrm{flag}}$ we have $\norm{\ip{X_{\tau}, 0^{\mathrm{flag}}|\widetilde{\Pi}U_q|0}}^2\geq (1-\epsilon)\mu(X_{\tau})$, where 
\[
    \ket{X_{\tau}} := \frac{1}{\sqrt{|X_{\tau}|}}\sum_{v\in X_{\tau}}\ket{v}
\]
is the uniform superposition over vectors in $X_{\tau}$. In particular, we can apply fixed-point amplitude amplification by \cref{thm:fpaa} to $U_q$, to obtain a boosted unitary $U_q'$ with $q'=O\left(\frac{1}{\sqrt{\mu(X_{\tau})}}\log\frac{1}{\delta}\right)$ uses of $U_q$ and $U_q^{\dagger}$, such that $\norm{\widetilde{\Pi}U_q'\ket0-\ket{X_{\tau},0^{\mathrm{flag}}}}\leq\delta$. Note that this procedure gives us a way to sample approximately from the uniform distribution on $X_{\tau}$ that is quadratically faster than a classical algorithm that draws random samples and verifies whether the $\ALG$ computes correctly on the drawn sample. Applying this argument to the case where $\tau=\alpha/2$, we have the following corollary.

\begin{corollary}
    \label{cor:quantum-sampling}
    Given an average-case quantum algorithm $\ALG$ with oracle access to the matrix $M\in\F^{n\times n}$ via $U_M$, there is a quantum algorithm $\mathsf{Q_{samp}}$ that uses $\ALG$ $O\left(\frac{1}{\sqrt{\tau\mu(X_{\tau})}}\right)$ times and with high probability outputs a vector $v\in X_{\tau}$ on which $\ALG$ succeeds with probability at least $\tau$. $\mathsf{Q_{samp}}$ makes $q=O\biggl(n^{3/2}\cdot\frac{1}{\sqrt{\tau\mu(X_{\tau})}}\biggr)$ queries to $U_M$ and $U_M^{\dagger}$, $O\biggl(q\log n\cdot\poly\log |\F|\biggr)$ additional one-qubit and two-qubit gates, and $O(n\log |\F|)$ ancillary qubits.
\end{corollary}

\subsection{Learning Bogolyubov subspaces from noisy quantum oracles}
\label{sec:quantumtools-fourier}

We would next like to design a quantum procedure for efficiently learning the Bogolyubov subspace from the noisy quantum oracle we constructed above in \cref{sec:indicator}. Since the Bogolyubov subspace is characterised by the Fourier coefficients of the indicator $\mathsf{1}_X$, our problem reduces to learning the Fourier spectrum of functions that are encoded in noisy quantum oracles of the aforementioned form.

\cite{AdcockCleve02} studied the problem of using noisy evaluations of a linear Boolean function $f(x)=a\cdot x$ to determine the underlying string $a\in\F_2^n$, given a guarantee that the evaluations have an average probability of at least $\frac12+\epsilon$ of being correct over random $x\in\F_2^n$. Here we show that the circuit that they consider can be used more generally to sample characters from the Fourier spectrum of a function $f:\F^n\to\F_2$ that is accessed via a noisy quantum oracle of the type defined in \cref{eq:bounded-error-ent-oracle}. In particular, we have the following result.
\begin{lemma}
\label{lem:qGoldLev}
    Suppose we are given as input a noisy quantum oracle $U_{f}$ to $f:\F^n\to\F_2$, such that
        \begin{equation}
        U_{f}\ket{x}\ket{0^{m+1}} = \ket{x}\biggl(\beta^x_{\mathrm{succ}}\ket{\mathrm{w}_0(x)}\ket{f(x)} + \beta^x_{\mathrm{fail}}\ket{\mathrm{w}_1(x)}\ket{\overline{f(x)}}\biggr),
    \end{equation}
    where $\forall x\in\F^n,~\abs{\beta^x_{\mathrm{succ}}}^2+\abs{\beta^x_{\mathrm{fail}}}^2 = 1$ and $\abs{\beta^x_{\mathrm{succ}}}^2 \geq 1-\epsilon$, $\overline{f(x)}=f(x)\oplus 1$ is $\mathsf{NOT}(f(x))$, and $\ket{\mathrm{w}_0(x)}$ and $\ket{\mathrm{w}_1(x)}$ are arbitrary $m$-qubit states of the workspace register. Then measuring the top three registers of the following circuit $\mathsf{C_{GL}}$ produces output $y\in\F^n$, $0^{m+1}$ with probability $p_y$ such that $\forall y\in\F^n$, $$\abs{|\hat{f}(y)|^2 - p_y} \leq 4\epsilon,$$ where $\hat{f}(y)$ are the Fourier coefficients of $f$, i.e.\ $f(x)=\sum_{y\in\F^n}\hat{f}(y)\chi_y(x)$ with $\chi_y(x):=\omega^{-x\cdot y}$.
    \vspace{1mm}    
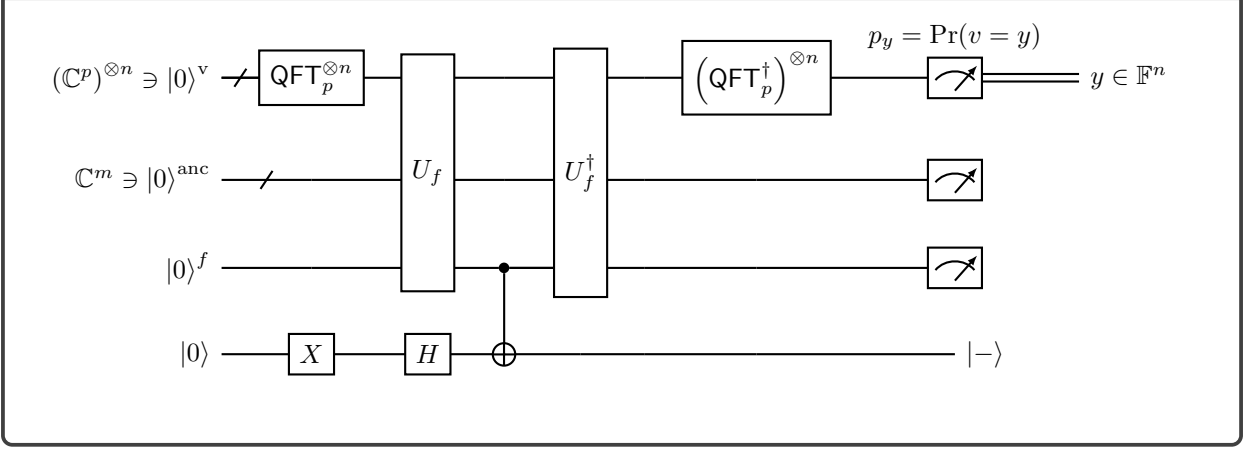
    \begin{tcolorbox}[title=Circuit~{$\mathsf{C_{GL}}$}: Quantum Goldreich-Levin algorithm for Fourier sampling,
        standard jigsaw,
        opacityback=0] 
    \small
        \begin{quantikz}
    \lstick{$\left(\C^{p}\right)^{\otimes n}\ni\ket{0}^{\mathrm{v}}$}  & \gate{\qft_p^{\otimes n}}\qwbundle{} & \gate[3]{U_{f}} & \qw & \gate[3]{U_{f}^{\dagger}} & \qw & \gate{\left(\qft_p^{\dagger}\right)^{\otimes n}} & \meter{$p_y=\Pr(v=y)$} & \cw \rstick{$y\in\F^n$}\\
    \lstick{$\C^{m}\ni\ket{0}^{\mathrm{anc}}$} &  \qwbundle{}                               &                              & \qw         &       & \qw  &\qw & \meter{} \\
    \lstick{$\ket{0}^f$}                       & \qw                                &                              & \ctrl{1}    & \qw   & \qw & \qw &\meter{} \\
    \lstick{$\ket{0}$}                         & \gate{X}                           & \gate{H}                     & \targ{}     & \qw   & \qw & \qw  & \qw\rstick{$\ket{-}$}\\
\end{quantikz}

    \end{tcolorbox}
    \vspace{2mm}    
\end{lemma}
Note that the number $m$ of ancillary qubits in the workspace register need not have any relation to $n$.
\begin{remark}
\label{rem:qGoldLev}
    In our application, we actually have a weaker assumption on the input $U_f$, which in fact is the indicator function of \cref{sec:indicator}. The guarantee we obtain on the Fourier sampler is then $$\abs{|\hat{f}(y)|^2 - p_y} \leq 4\epsilon + 4\rho_W,$$ where $\rho_W:=|W|/|\F|^n$ is the density of the intermediate set $W$.
\end{remark}
\begin{proof}
    The input state undergoes the following transformations through the circuit $\mathsf{C_{GL}}$ : 
    \begin{align}
    \label{eq:noisyBVfirstHalf}
        \ket{0}\ket{0}\ket{0}&\ket{0} \xrightarrow{\qft_p^{\otimes n}\otimes\mathbbm{1}\otimes\mathbbm{1}\otimes X}\frac{1}{\sqrt{N}}\sum_{x\in\F^n}\ket{x}\ket{0}\ket{0}\ket{1}  \\
            &\xrightarrow{~~~U_{f}\otimes H~~~}\frac{1}{\sqrt{N}}\sum_{x\in\F^n}\ket{x}\biggl(\beta^x_{\mathrm{succ}}\ket{\mathrm{w}_0(x)}\ket{f(x)} + \beta^x_{\mathrm{fail}}\ket{\mathrm{w}_1(x)}\ket{\overline{f(x)}}\biggr)\ket{-} \nonumber\\
            &\xrightarrow{\mathbbm{1}\otimes\mathbbm{1}\otimes\mathsf{CNOT}}\frac{1}{\sqrt{N}}\sum_{x\in\F^n}\ket{x}\biggl((-1)^{f(x)}\beta^x_{\mathrm{succ}}\ket{\mathrm{w}_0(x)}\ket{f(x)} + (-1)^{\overline{f(x)}}\beta^x_{\mathrm{fail}}\ket{\mathrm{w}_1(x)}\ket{\overline{f(x)}}\biggr)\ket{-}\nonumber\\
            &=\frac{1}{\sqrt{N}}\sum_{x\in\F^n}(-1)^{f(x)}\ket{x}\biggl(\beta^x_{\mathrm{succ}}\ket{\mathrm{w}_0(x)}\ket{f(x)} - \beta^x_{\mathrm{fail}}\ket{\mathrm{w}_1(x)}\ket{\overline{f(x)}}\biggr)\ket{-},\nonumber
    \end{align}
    where $N=p^n$. In order to establish that the output is $y\in\F^n$ with probability $p_y$ such that $|p_y-\hat{f}(y)^2|<4\epsilon$ on measuring the first register after executing the entire circuit, we will compute the inner product 
    \begin{equation}
    \label{eq:noisyBVamps}
        \beta_y:=\biggl\langle y,0,0,-\biggl|\mathrm{C}\biggr|0,0,0,0\biggr\rangle,
    \end{equation}
    since $p_y=\abs{\beta_y}^2$. Notice that we also have
    \begin{align}
    \label{eq:noisyBVsecondHalf}
        \ket{y,0,0,-}&\xrightarrow{\qft_p^{\otimes n}\otimes\mathbbm{1}\otimes\mathbbm{1}\otimes \mathbbm{1}}\frac{1}{\sqrt{N}}\sum_{v\in\F^n}\omega^{y\cdot v}\ket{v,0,0,-}\\
            &\xrightarrow{U_{f}\otimes\mathbbm{1}}\frac{1}{\sqrt{N}}\sum_{v\in\F^n}\omega^{y\cdot v}\ket{v}\biggl(\beta^v_{\mathrm{succ}}\ket{\mathrm{w}_0(v)}\ket{f(v)} + \beta^v_{\mathrm{fail}}\ket{\mathrm{w}_1(v)}\ket{\overline{f(v)}}\biggr)\ket{-}\nonumber,
    \end{align}
    and we can compute the amplitude $\beta_y$ in \cref{eq:noisyBVamps} by taking the inner product between the states on the last lines of \cref{eq:noisyBVfirstHalf,eq:noisyBVsecondHalf}. Since $f(x)\in\F_2$ we have that $\ip{f(x)|\overline{f(x)}}=0$. Similarly, $\ip{v|x}=\delta_{vx}$, and while the states of the workspace register may not be orthogonal, they are normalised. Hence we have
    \begin{align}
        \beta_y &= \frac1N \sum_{v\in\F^n} \omega^{y\cdot v}(-1)^{f(v)} \biggl(\abs{\beta^v_{\mathrm{succ}}}^2 - \abs{\beta^v_{\mathrm{fail}}}^2\biggr).
    \end{align}
    Since for every $v\in\F^n$, $\abs{\beta^v_{\mathrm{succ}}}^2\geq 1-\epsilon$ and $\abs{\beta^v_{\mathrm{succ}}}^2+\abs{\beta^v_{\mathrm{succ}}}^2=1$, it always holds that $\abs{\beta^v_{\mathrm{succ}}}^2-\abs{\beta^v_{\mathrm{succ}}}^2\geq 1-2\epsilon$. Unlike the Boolean case (i.e.\ $\F=\F_2$), the Fourier coefficients of $f$ are no longer real numbers, and so need some additional care. Nevertheless, using their definition from \cref{eq:fourier-coefts}, and the Cauchy-Schwarz inequality, we have
    \begin{align}
    \label{eq:nosyBV-end}
       \abs{\hat{f}(y)-\beta_y} &= \frac1N \abs{\sum_{v\in\F^n} \omega^{y\cdot v}(-1)^{f(v)} \biggl(1-\biggl(\abs{\beta^v_{\mathrm{succ}}}^2 - \abs{\beta^v_{\mathrm{fail}}}^2\biggr)\biggr)}\nonumber\\
            &\leq \frac1N \abs{\sum_{v\in\F^n} \omega^{2y\cdot v}(-1)^{2f(v)}}^{1/2} \abs{\sum_{v\in\F^n}\biggl(1-\biggl(\abs{\beta^v_{\mathrm{succ}}}^2 - \abs{\beta^v_{\mathrm{fail}}}^2\biggr)\biggr)^2}^{1/2}\nonumber\\
            &\leq \frac1N\cdot\sqrt{N}\cdot 2\epsilon\sqrt{N}\nonumber\\
            &\leq 2\epsilon.
    \end{align}
    As $|\hat{f}(y)|,\abs{\beta_y}\leq 1$, and since $\abs{\abs{x}^2-\abs{y}^2}\leq\abs{x^*+y^*}\abs{x-y}\leq 2\abs{x-y}$ when $0\leq |x|,|y| \leq 1$ and $x^*$ denotes complex conjugation, we finally have that 
    \begin{align}
        \abs{|\hat{f}(y)|^2-p_y} &\leq 4\epsilon,
    \end{align}
    showing that the circuit $\mathsf{C_{GL}}$ can sample approximately from the Fourier spectrum of $f$. 
\end{proof}
The quantum Fourier transform $\qft_p$ over finite fields $\F_p$ can be implemented efficiently with $O(\log^2 |\F|)$ elementary gates and in depth $O(\log |\F|)$ \cite{Beals97QFT,Hoyer97QFT,Moore06QFT}. 

\paragraph{Learning subspace coefficients of $\mathsf{1_X}$:} Our interest is in the indicator function $\mathsf{1}_{X}:\F^n\to\{0,1\}$, which takes value $1$ on the set of vectors on which $\ALG$ succeeds with appreciable probability $|\beta_{\mathrm{succ}}^v|^2\geq \alpha$. 
Using the unitary $U_q$ of \cref{sec:indicator} almost satisfies the condition required for \cref{lem:qGoldLev}: $\abs{\beta_{\mathrm{succ}}}^2\geq 1-\epsilon$ and $\abs{\beta_{\mathrm{fail}}}^2 \leq \epsilon$, where $\epsilon$ can be chosen to be $o(1)$. To see the bound noted in \cref{rem:qGoldLev}, suppose the densities of the sets $X^g_t,X^b_t$ and $W_t$ are $\rho_g,\rho_b$ and $\rho_W$ respectively and let $\pi_v:=\biggl(1-\biggl(\abs{\beta^v_{\mathrm{succ}}}^2 - \abs{\beta^v_{\mathrm{fail}}}^2\biggr)\biggr)^2$. Then in \cref{eq:nosyBV-end} we have
\begin{align}
       \abs{\hat{f}(y)-\beta_y} &\leq \frac1N \abs{\sum_{v\in\F^n} \omega^{2y\cdot v}(-1)^{2f(v)}}^{1/2} \abs{\sum_{v\in\F^n}\biggl(1-\biggl(\abs{\beta^v_{\mathrm{succ}}}^2 - \abs{\beta^v_{\mathrm{fail}}}^2\biggr)\biggr)^2}^{1/2}\nonumber\\
        &\leq \frac1N\cdot\sqrt{N}\cdot \abs{\sum_{v\in X^g_t}\pi_v + \sum_{v\in X^b_t}\pi_v + \sum_{v\in W_t}\pi_v}^{1/2}\nonumber\\
        &\leq \frac{1}{\sqrt{N}}\cdot \abs{2\epsilon\rho_g N + 2\epsilon\rho_b N + 2\rho_W N}^{1/2}\nonumber\\
            &\leq 2\epsilon + 2\rho_W,
\end{align}
where on the third line we used the definition of $U_q$ to see that the quantity $\pi_v\leq 2\epsilon$ on both $X^g_t$ and $X^b_t$, while we only have the guarantee that $\pi_v\leq 2$ on the set $W_t$, and that $\rho_g+\rho_b+\rho_W=1$. Hence we have the following corollary.    

\begin{corollary}
    \label{lem:q-fourier-sampling}
    Given an average-case quantum algorithm $\ALG^M$ as described in \cref{eq:BEQLA} and the oracle $U_M$ for a matrix $M\in\F^{n\times n}$, we can learn heavy Fourier characters $\chi_y$ with $\widehat{f}(y)\geq c$ of the indicator function on the set $X$ of inputs on which $\ALG$ succeeds with high probability, using $\ALG$ $O(1/c)$ times, $q=O\left(\frac{1}{c}\cdot n^{3/2}\right)$ queries to $U_M$ and $U_M^{\dagger}$, $O(q \cdot\log n\cdot \poly\log |\F|)$ additional one-qubit and two-qubit gates, and $O(n\log|\F|)$ ancillary qubits.
\end{corollary}

The procedure is to simply measure the state output by $\mathsf{C_{GL}}$ to obtain a character $y\in\F^n$, and then run standard quantum amplitude estimation to estimate the value of the corresponding Fourier coefficient to additive precision $c/100$, which will use $\mathsf{C_{GL}}$ a total of $O(1/c)$ times. 

\section{Robust quantum local correction via additive combinatorics}
\label{sec:AC}

In this section, we prove our main technical tool: a robust quantum local correction lemma for linear problems. Towards this end, we first prove a noise-robust generalisation of Bogolyubov's lemma from additive combinatorics.

\subsection{Robust probabilistic Bogolyubov lemma}

 Recall that Bogolyubov's lemma states that for any subset $A \subseteq \F_2^n$ of density $|A|/2^n \geq \alpha$, there exists a subspace $V \subseteq 4A$ of dimension at least $n - \alpha^{-2}$.
 
 We would like to use Bogolyubov's lemma to locally correcting faulty inputs by \emph{explicitly} computing a decomposition into a linear combination of \textit{good inputs}, shifted by a sparse vector (see \cref{tech:qcorrection}). Computing the sparse shift vector requires learning the (significant) Fourier spectrum of the subspace implied by Bogolyubov's lemma. The caveat is that our quantum algorithm for learning the subspace encoded in a noisy quantum state (see \cref{sec:quantumtools-fourier}) can only obtain an approximation of the spectrum.
 
 Hence, we need to strengthen Bogolyubov's lemma to obtain the following structural properties:
 \begin{enumerate}
     \item \emph{Robustness}, in the sense that the linear constraints that define the Bogolyubov subspace can lie within a range of Fourier thresholds; and
     \item \emph{Density}, in the sense that each element of the Bogolyubov subspace admits many decompositions into valid inputs, and in turn can be sampled probabilistically using our quantum sampling procedure (see \cref{sec:quantum-sampling}).
 \end{enumerate}


Next, we prove a generalisation of Bogolyubov's lemma, which achieves the aforementioned structural properties. In the following, given a set $X \seq \F^n$ define $\Spec_X(\gamma) = \{r \in \F^n \setminus \{0\} : \abs{\hat{1}_X(r)} \geq \gamma\}$.

\begin{lemma}[Robust Bogolyubov lemma]\label{lem:apx-bogolyubov}
    Let $\F = \F_p$ be a prime field, and 
    let $X \seq \F^n$ be a set of size $\abs{X} = \alpha \cdot \abs{\F}^n$ for some $\alpha \in (0,1]$.
    Let $R \seq \F^n$ be a set such that
    $\Spec_X(\alpha^{3/2}) \seq R \seq \Spec_X(\frac{\alpha^{3/2}}{2})$.
    
    Let $V = \{v \in \F^n : \ip{v,r} = 0 \ \forall r \in R\}$.
    Then $\dim(V) \geq n-4/\alpha^2$, and for all $v \in V$ it holds that
    \begin{equation*}
        \Pr_{x_1,x_2,x_3 \in \F^n}[x_1,x_2,x_3,v - x_1 - x_2 - x_3 \in X] \geq \alpha^5 \;,
    \end{equation*}
    or equivalently
    \begin{equation*}
        \Pr_{x_1,x_2,x_3 \in X}[v - x_1 - x_2 - x_3 \in X] \geq \alpha^2 \;.
    \end{equation*}
\end{lemma}

\begin{proof}
    Note first that by Parseval's identity we have
    \begin{equation*}
            \alpha = \ip{1_X,1_X} = \norm{1_X}_2^2 = \sum_{r} \abs{\hat{1}_X(r)}^2 \;.
    \end{equation*}
    In~particular,
    \begin{equation*}
        \frac{\alpha^3}{4} \cdot |R| \leq \sum_{r \in R} \abs{\hat{1}_X(r)}^2 \leq \sum_{r} \abs{\hat{1}_X(r)}^2 = \alpha;,
    \end{equation*}
    and hence $|R| \leq \frac{4}{\alpha^2}$. In particular, $\dim(V) \geq n - |R| \geq n - 4/\alpha^2$.
    
    Furthermore, we have
    \begin{equation*}
      \sum_{r \in \F^n \setminus (R \cup \{0\})} \abs{\hat{1}_X(r)}^4
      \leq \alpha^3 \cdot \sum_{r \in \F^n \setminus (R \cup \{0\})} \abs{\hat{1}_X(r)}^2
      \leq \alpha^3 (\alpha - \alpha^2) \leq \alpha^4 - \alpha^5
      \enspace,
    \end{equation*}
    where the second inequality uses that $\sum_r \abs{\hat{1}_X(r)}^2=\alpha$, and $\abs{\hat{1}_X(0)}^2=\alpha^2$.
    Noting that for every $v \in V$ we have $\chi_r(v) = \omega^{\ip{v,r}} = \omega^0 = 1$ for all $r \in R$,
    it follows that
    \begin{eqnarray*}
        \Pr_{x_1,x_2,x_3 \in \F^n}[x_1,x_2,x_3,v-x_1-x_2-x_3 \in V]
        & = & (1_X * 1_X * 1_X * 1_X)(v) \\
        & = & \sum_{r \in \F^n} (\hat{1}_X(r))^4 \chi_r(v) \\
        & = & \abs{\hat{1}_X(0)}^4 \chi_0(v) + \sum_{r \in R} \abs{\hat{1}_X(r)}^4 \chi_r(v) \\
        && + \sum_{r \in \F^n \setminus (R \cup \{0\})} \abs{\hat{1}_X(r)}^4 \chi_r(v) \\
        & \geq & \alpha^4 +  (\alpha^{3/2}/2)^4 \cdot \abs{R} - (\alpha^4 - \alpha^5) \\
        & \geq & \alpha^5
        \enspace,
    \end{eqnarray*}
    as required.
\end{proof}

\subsection{Quantum local correction lemma}

Equipped with the noise-robust Bogolyubov lemma, we can now employ the four quantum procedures we showed in \cref{sec:quantumtools} and prove the following quantum local correction lemma for linear problems, which lies at the heart of our worst-case to average-case reductions.

\begin{lemma}[Quantum local correction]\label{lem:quantum-local-correction}
    Let $\F = \F_p$ be a prime field, and let $X \seq \F^n$ be a set of size $\abs{X} = \alpha \cdot \abs{\F}^n$, for some $\alpha \in (0,1]$. Then, there exists a non-negative integer $t \leq 4/\alpha^2$, a~collection of $t$ vectors $B = \{b_1,\dots,b_t \in \F^n\}$, and $t$ indices $k_1,\dots, k_t \in [n]$ satisfying the following:
    \begin{description}
      \item Given a vector $y \in \F^n$, define $s = \sum_{j=1}^t \ip{y, b_j} \cdot \vec{e}_{k_j}$
            where $(\vec{e}_i)_{i \in [n]}$ is the standard basis.
            Then
            \begin{equation*}
                \Pr_{x_1,x_2,x_3 \in \F^n}[x_1,x_2,x_3,x_4 \in X] \geq \alpha^5
                \enspace,
            \end{equation*}
    \end{description}
    where $x_4=y-x_1-x_2-x_3-s$.
    Equivalently, we have
            \begin{equation*}
                \Pr_{x_1,x_2,x_3 \in X}[x_4  \in X] \geq \alpha^2
                \enspace.
            \end{equation*}
    
    Furthermore, suppose we have a quantum membership oracle $\widetilde{O}_{X^*}$ (that we can query in superposition) for a set $X^*$ satisfying the following conditions
    \begin{itemize}
        \item $\abs{\hat{1}_{X^*}(r) - \hat{1}_{X}(r)} \leq \frac{\alpha^{3/2}}{8}$
        for all $r \in \F^n$;
        \item for every input $x \in \F^n$, $\widetilde{O}_{X^*}$ computes the indicator $1_{X^*}(x)$ correctly with probability at least $1-\alpha^{3/2}/10$.
    \end{itemize}
    Then, there exists a quantum algorithm that with probability at least $1-\delta$ returns vectors $b_1,\dots,b_t \in \F^n$ and indices $k_1,\dots, k_t \in [n]$ as described above. This algorithm makes $O( \alpha^{-7/2})$ blackbox queries to $\widetilde{O}_{X^*}$, uses an additional number
    \begin{align*}
        O\biggl(n\alpha^{-7/2}\cdot\log n\cdot\poly\log|\F|\biggr) \enspace
    \end{align*}
    of one-qubit and two-qubit gates, and $O(n\alpha^{-2}\log|\F|)$ ancillary qubits.
\end{lemma}

\begin{proof}
    Fix a set $X \seq \F^n$ of size $\abs{X} = \alpha \cdot \abs{\F}^n$ for some $\alpha \in (0,1]$.
    By applying \cref{lem:apx-bogolyubov}, we obtain a subspace $V \seq \F^n$ of dimension $\dim(V) = n - t$ for $t = 4/\alpha^2$.
    Let $R \subseteq \F_2^n \setminus \{0\}$ be a set of vectors in $\F^n$ of size $t$ such that $V = \{v \in \F_2^n : \ip{v,r} = 0 \; \forall r \in R\}$. Indeed, we can let $R$ be a set of $t$ linearly independent vectors in $V^\perp$.
    
    By writing the vectors of $R$ in a matrix and diagonalizing the matrix,
    we obtain: (1) a set of vectors $B = \{b_1,\dots, b_t \in \F_2^n\}$ such that $\sp(B) = \sp(R)$,
    and (2) the corresponding pivot indices $k_1,\dots,k_t \in [n]$
    such that $b_j[k_j] = 1$ and $b_j[k_{j'}] = 0$ for all $j \neq j'$.

    Given a vector $y \in \F^n$, define $s = \sum_{j=1}^t \ip{y, b_j} \cdot \vec{e}_{k_j}$, where $(\vec{e}_i)_{i \in [n]}$ is the standard basis,
    and let $v = y-s$.
    It is straightforward to verify that $v \in V$.
    Then for any $j \in [t]$ we have
    \begin{equation*}
        \ip{v, b_j}
        = \ip{y, b_j}  - \sum_{j=1}^t c_j \cdot \ip{\vec{e}_{k_j},b_j}
        \overset{\text{(*)}}{=} \ip{y, b_j}  - c_j \cdot \ip{\vec{e}_{k_j},b_j}
        \overset{\text{(**)}}{=} \ip{y, b_j}  - \ip{y, b_j} = 0
        \enspace,
    \end{equation*}
    where (*) is because $\ip{\vec{e}_{k_{j'}},b_j} = b_j[i_{j'}] = 0$ for $j \neq j'$,
    and (**) is because $\ip{\vec{e}_{k_j},b_j} = b_j[i_{j}] = 1$.
    
    Now, since $v \in V$, by the guarantees of \cref{lem:apx-bogolyubov} it follows that 
    \begin{equation*}
        \Pr_{x_1,x_2,x_3 \in \F^n}[ x_1 \in X, x_2 \in X, x_3 \in X,v - x_1 - x_2 - x_3 \in X] \geq \alpha^5
        \enspace,
    \end{equation*}
    which is equivalent to 
    \begin{equation*}
        \Pr_{x_1,x_2,x_3 \in X}[v - x_1 - x_2 - x_3 \in X] \geq \alpha^2
        \enspace.
    \end{equation*}
    
    For the furthermore part, consider the oracle $\widetilde{O}_{X^*}(x)$.
    By the closeness between the Fourier coefficients of $1_X$ and $1_{X^*}$, we may apply \cref{lem:apx-bogolyubov}. Indeed, letting $R = \Spec_{X^*}(\frac34\alpha^{3/2})$ we have a set of Fourier coefficients satisfying the requirements of \cref{lem:apx-bogolyubov}, namely that $\Spec_{X^*}(\alpha^{3/2}) \seq R \seq \Spec_{X^*}(\alpha^{3/2}/2)$.
    
    By definition, for every input $x \in \F^n$, $\widetilde{O}_{X^*}(x)$ computes the indicator $1_{X^*}(x)$ correctly with probability at least $1-\alpha^{3/2}/10$. Now, we can apply \cref{lem:q-fourier-sampling} to find a $y\in \F^n$ such that
    $\hat{1}_{X^*}(y) \geq 3\alpha^{3/2}/4$ with $O(\alpha^{-3/2})$ blackbox queries to $\widetilde{O}_{X^*}$, and $O(n\alpha^{-3/2}\cdot\poly\log |\F|)$ additional one-qubit and two-qubit gates. 
    By the closeness between the Fourier coefficients of $1_X$ and $1_{X^*}$, these $y$'s also satisfy $\hat{1}_X(y) \geq \alpha^{3/2}/2$.
    Since there are at most $t\leq 4/\alpha^2$ such coefficients, to find all of them with probability $1-\delta$, we need to repeat the sampling procedure above $O(\alpha^{-2})$ times. Hence, the total query complexity is $O(\alpha^{-7/2})$ blackbox queries to $\widetilde{O}_{X^*}$, the total number of additional gates is
    \begin{align*}
        O\biggl(n\alpha^{-7/2}\cdot\log n\cdot\poly\log|\F|\biggr) \enspace,
    \end{align*}
    and the number of ancillary qubits is $O(n\cdot\alpha^{-2}\log|\F|)$.
\end{proof}

\section{Reductions for linear problems}
\label{sec:reduction}
In this section, we prove \cref{lem:main}. Namely, we will present an algorithm $\ALG'$ whose complexity is as required by the lemma statement. In order to prove the correctness of the algorithm, we will consider $M\in\F^{n\times n}$ such that $\Pr_{v,\ALG}[\ALG^M(v) = Mv]\geq\alpha$, and we will prove for this $M$ and every vector $v\in\F^n$, that  $\Pr_{\ALG'}[(\ALG')^M(v) = Mv]\geq 1-\delta$. To this end we fix a matrix $M$ satisfying the premise: 
\begin{align}\label{eq:goodM}
\Pr_{v,\ALG}[\ALG^M(v) = Mv]\geq\alpha \;.
\end{align}

Before describing the proof of \cref{lem:main} we introduce the following notation. For each $v \in \F^n$ let $p_v = \Pr_{\ALG}[\ALG(v) = Mv]$ be the probability that $\ALG$ computes correctly the output on input $v$, where the probability is taken  only over the quantum randomness of $\ALG$.

Before proceeding with the proof, we need the following definition of threshold sets.
    \begin{definition}\label{def:set-X}
        For an algorithm $\ALG^M$ and a matrix $M$ satisfying~\cref{eq:goodM}, we define the set of its \textit{good} inputs, i.e., the inputs $v\in\F^n$ such that $\ALG^M(v)=M\cdot v$ with a non-negligible probability as follows. 
        \begin{align*}
            X_{\kappa}\colon= \left\{v \in \F^n \colon p_v \geq {\kappa} \right\} ;.
        \end{align*}
    \end{definition}
    
\subsection{Properties of threshold sets}
\label{subsec:thresholds}
Below we prove several claims regarding the threshold sets $X_\kappa$.
    \begin{claim}\label{clm:density-of-X}
    For $\kappa \leq {\alpha}/2$, the density of $X_{\kappa}$ is at least $\alpha/2$, i.e., $\abs{X_{\kappa}} \geq \frac{\alpha}{2}\abs{\F}^n$.
    \end{claim}
    \begin{proof}
    By the assumption of the lemma we have $\E_{v\in \F^n}[p_v] \geq \alpha$, and hence
    \begin{eqnarray*}
        \alpha  \leq  \E_v[p_v]
        \leq 1 \cdot \Pr_v[p_v \geq \alpha/2] + \alpha/2 \cdot \Pr_v[p_v < \alpha/2] 
        \leq  \Pr_v[p_v \geq \kappa] + \alpha/2 \;.
    \end{eqnarray*}
    Therefore, for all $\kappa \leq \alpha/2$, we have that 
    $\Pr_v[p_v \geq \kappa] \geq \alpha/2$, as required.
    \end{proof}
    
    Next, we prove that if we choose a random $\tau \in [\alpha/4, \alpha/2]$ and $\tau' = \tau - 1/K$ for a sufficiently large $K$,
    then, with high probability over the random choices of $\tau$ the sets $X_{\tau}$ and $X_{\tau'}$ will have almost the same density.

    \begin{claim}\label{clm:close-sets}
        For a parameter $K$
        let $r \in \{1,\dots, K\}$ be chosen uniformly at random.
        Let $\tau = (1+r/K) \frac{\alpha}{4}\in [\frac{\alpha}{4}+\frac{\alpha}{4K}, \frac{\alpha}{2}]$, and $\tau' = \tau - \frac{\alpha}{4K}$.
        Then, $\Pr[\frac{\abs{X_{\tau'}}}{\abs{\F^n}} - \frac{\abs{X_{\tau}}}{\abs{\F^n}} \leq \frac{2}{K}] > 1/2$.
    \end{claim}
    \begin{proof}
        Note that the interval $[\tau,\tau']$ is sampled by dividing $[\alpha/4,\alpha/2]$ into $K$ intervals, and taking one of them uniformly at random, it follows that
        $\E\left[\frac{\abs{X_{\tau'}}}{\abs{\F^n}} - \frac{\abs{X_{\tau}}}{\abs{\F^n}}\right] \leq 1/K$.
        The claim follows by Markov's inequality.
    \end{proof}
    
    In particular \cref{clm:close-sets} implies the following corollary.
    \begin{corollary}\label{cor:close-fourier-coeffs}
        For a parameter $K$
        let $r \in \{1,\dots, K\}$ be chosen uniformly at random.
        Let $\tau = (1+r/K) \frac{\alpha}{4}\in [\frac{\alpha}{4}+\frac{\alpha}{4K}, \frac{\alpha}{2}]$, and $\tau' = \tau - \frac{\alpha}{4K}$.
        
        Suppose that $X^* \seq \F^n$ is an arbitrary set such that $X_{\tau} \seq X^* \seq X_{\tau'}$. Then with probability at least $1/2$ (over the choice of $\tau$) it holds that 
        \begin{equation*}
            \abs{\hat{1}_{X^*}(r) - \hat{1}_{X_{\tau}}(r)} \leq \frac{2}{K}
        \end{equation*}
        for all $r \in \F^n$.
    \end{corollary}
    
    \begin{proof}
        By \cref{clm:close-sets} we have $\Pr[\frac{\abs{X_{\tau'}}}{\abs{\F^n}} - \frac{\abs{X_{\tau}}}{\abs{\F^n}} \leq \frac{2}{K}] > 1/2$.
        Suppose this event happens and $\frac{\abs{X_{\tau'}}}{\abs{\F^n}} - \frac{\abs{X_{\tau}}}{\abs{\F^n}} \leq \frac{2}{K}$.
        Since $X_{\tau} \seq X^* \seq X_{\tau'}$, it follows that
        $\frac{\abs{X^*}}{\abs{\F^n}} - \frac{\abs{X_{\tau}}}{\abs{\F^n}} \leq \frac{2}{K}$.
        We observe that 
        \begin{equation*}
            \abs{\hat{1}_{X^*}(r) - \hat{1}_{X_{\tau}}(r)} \leq \abs{\frac{\abs{X_{\tau}}}{\abs{\F^n}} - \frac{\abs{X^*}}{\abs{\F^n}}} \leq \frac{2}{K} \;,
        \end{equation*}
        which finishes the proof of the corollary.
    \end{proof}

\subsection{Proof of Lemma~\ref{lem:main}}

In this section, we prove \cref{lem:main}, which we restate below for convenience.

\lemmain*

Let $\ALG^M$ be the average-case quantum algorithm from the lemma statement that has query access to the matrix $M$. 
Note that given the algorithm $\ALG^M$ we can define an algorithm $\ALG^M_{boost}$ that given an input $v$ makes $O(1/\alpha)$ calls to $\ALG^M(v)$ independently, and each time verifies the result using \cref{lem:qVerification}. Therefore, for every constant $\delta$, we may assume that there is at least an $\alpha/2 $ fraction of inputs on which $\ALG^M_{boost}$ outputs the correct answer with probability at least $1-\delta/8$. 

We define $(\ALG')^M$ as follows.

\begin{tcolorbox}[title=Algorithm~{1}: Reduction for linear problems] 
\paragraph{Input: $\ALG^M$, $v \in \F^{n}$, $\alpha \in (0,1)$}
\paragraph{Output: $M \cdot v$}
\begin{enumerate}
    \item Let $K = \frac{4}{\alpha^{3/2}}$, and let $r \in \{1,\dots, K\}$ be chosen uniformly at random.

    \item  \textsf{Setting a random threshold:} Let $\tau = (1+r/K) \frac{\alpha}{4}\in [\frac{\alpha}{4}, \frac{\alpha}{2}-\frac{1}{K}]$, and $\tau' = \tau - \frac{1}{K}$.

    \item \label{item:create-oracle-X*}\textsf{Membership Oracle for $X^*$:} Construct a quantum noisy membership oracle $\widetilde{O}_{X^*}$ using \cref{sec:indicator} such that with high probability $X_{\tau} \seq X^* \seq X_{\tau'}$.

    \item \label{item:learning-X*}\textsf{Learning $X_\tau$:} Having $\widetilde{O}_{X^*}$ defined above, apply the quantum local correction lemma (\cref{lem:quantum-local-correction}) with error parameter $\delta=1/6$ to compute a collection of $t = O(1/\alpha^2)$ vectors $B = \{b_1,\dots,b_t \in \F^n\}$ and indices $k_1,k_2,...,k_t$.
    Indeed, by \cref{cor:close-fourier-coeffs}, we may apply \cref{lem:quantum-local-correction}.

    \item  \label{item:sample-from-X*} \textsf{Efficient sampling from $X_\tau$:} Sample $x_1, x_2, x_3$ from $X_{\tau}$ using \cref{cor:quantum-sampling} with $\kappa=\tau$. 
    
    
    \item \label{item:self-correcting-X*}\textsf{Self-correcting $v$:} Using the set $B$ and indices $k_i$ computed in Step~3, let $s$ be the $t$-sparse vector defined as in \cref{lem:quantum-local-correction}:
    \begin{align*}s &= \sum_{j=1}^t \ip{v, b_j} \cdot \vec{e}_{k_j}\;, \text{ and}\\
    x_4 &= v - x_1 - x_2 - x_3 - s \;.
    \end{align*}

    \item \label{item:computing-Mv}\textsf{Computing $M\cdot v$:} Let $b = \ALG_{boost}^M(x_1) + \ALG^M_{boost}(x_2) + \ALG^M_{boost}(x_3)
    + \ALG^M_{boost}(x_4) + M\cdot s$, where $M\cdot s$ is computed by multiplying $M$ by all~$t$ non-zero entries of~$s$.
    
    \item \label{item:verify}\textsf{Verification:} Using \cref{lem:qVerify-simple} with $\eps=O(\alpha^2)$ verify whether $b = M \cdot v$. If the verification accepts, return $b$.
    
    \item \textsf{Probability amplification:} Repeat the steps~\ref{item:sample-from-X*}--\ref{item:verify} $O(1/\alpha^2)$ times.
\end{enumerate}
\end{tcolorbox}

    \paragraph{Correctness:} By our choice of the parameters and the justifications in the algorithm, we see that with constant probability we simultaneously have that
    (i) the quantum oracle $\widetilde{O}_{X^*}$ from \cref{lem:qVerification} for every $x \in \F^n$, outputs $\widetilde{O}_{X^*}(x)=1_{X^*}(x)$ in Step~\ref{item:create-oracle-X*}; (ii) The set $B$ and indices $k_i$ are computed correctly by \cref{lem:quantum-local-correction} in Step~\ref{item:learning-X*};
    (iii) Step~\ref{item:sample-from-X*} produces uniformly random samples $x_1,x_2,x_3 \in X_\tau$.
    
    Assuming successful executions of the previous steps, by \cref{lem:quantum-local-correction} for $x_4$ computed in Step~\ref{item:self-correcting-X*} we have that
    \begin{align*}
        \Pr_{x_1,x_2,x_3 \in X_{\tau}}[x_4 \in X_{\tau}] \geq \alpha^2 \;.
    \end{align*}
    Observing that for $x_1, x_2, x_3,x_4 \in X^*$ it holds that $\Pr[\ALG^M_{boost}(x_i) = Mx_i \;\forall i=1,2,3,4]\geq 1-\delta/2$.
    Therefore, with probability $\Omega(\alpha^2)$, the result computed in Step~\ref{item:computing-Mv} is the correct output $b=M\cdot v$.
    
    Given the representation of $X_\tau$ from Step~\ref{item:learning-X*}
    we repeat steps~\ref{item:sample-from-X*}--\ref{item:verify} $O(1/\alpha^2)$ times to amplify the probability of success to $2/3$ for every input vector~$v$.

    \paragraph{Query complexity:} 
    First, we bound the number of queries made by our algorithm in each iteration (Steps 3--8). This amounts to bounding the number of queries required to learn the heavy Fourier coefficients of $1_{X^*}$ in Step~\ref{item:learning-X*} (using the oracle from Step~\ref{item:create-oracle-X*}), queries in Step~\ref{item:sample-from-X*} to sample from $X_M$, queries in the four executions of $\ALG^M_{boost}$ in Step~\ref{item:computing-Mv} (each corresponding to $O(1/\alpha)$ calls to $\ALG^M$ and $O(1/\alpha)$ times verifying the computation), queries required to compute $M \cdot s$ in Step~\ref{item:computing-Mv}, and queries required for verification in Step~\ref{item:verify}. 
    
    By \cref{lem:quantum-local-correction}, in Step~\ref{item:learning-X*} we need at most $O(\alpha^{-7/2})$ queries to $\ALG^M$, and $O(n^{3/2}\cdot \alpha^{-7/2})$ queries to $U_M$ and $U_M^{\dagger}$
    to find the set $B$ and indices $k_1,k_2,...,k_t$. Note that the algorithm does not repeat Step~\ref{item:learning-X*} $O(\alpha^{-2})$ times.
    
    In step~\ref{item:sample-from-X*}, by \cref{cor:quantum-sampling}, we make $O(1)$ uses of the noisy unitary implementation of the approximate indicator function in \cref{sec:indicator}, which translates to $O(n^{3/2}\cdot\alpha^{-3/2})$ queries to $U_M$ and $U_M^{\dagger}$ and $O(\alpha^{-3/2})$ queries to $\ALG^M$. 
    In Step~\ref{item:computing-Mv}, each application of $\ALG^M_{boost}$ uses $O(1/\alpha)$ calls to $\ALG^M$ and  $O(1/\alpha)$ verification steps.
    In addition $s$ is $t$-sparse for $t=4/\alpha^2$, we compute $M \cdot s$ by making at most $O(n/\alpha^2)$ queries.
    Step~\ref{item:verify} makes $O(n^{3/2})$ queries to $U_M$ and $U_M^{\dagger}$.
    Repeating steps~\ref{item:sample-from-X*}--\ref{item:verify} for $O(1/\alpha^2)$ times leads to the bound of $O(n^{3/2}\cdot\alpha^{-7/2})$ on the query complexity of the constructed algorithm.

    \paragraph{Efficiency:} Now we count the additional gates used by our algorithm. The membership oracle from \cref{lem:qVerification} in Step~\ref{item:create-oracle-X*}, Fourier sampling circuit from \cref{lem:quantum-local-correction} in Step~\ref{item:learning-X*}, and efficient sampling circuit from \cref{cor:quantum-sampling} in Step~\ref{item:sample-from-X*} all use $O(n^{3/2}\log{n}\cdot \alpha^{-3/2} \cdot \poly(\log\abs{\F}))$ additional gates. Additionally, we need $O(n)$ gates to represent vectors $s$ and $v_i$ for $i \in \{1,2,3,4\}$.
    Thus, in total our algorithm uses $O(n^{3/2}\log{n}\cdot \alpha^{-7/2} \cdot \poly(\log\abs{\F}))$ additional one-qubit and two-qubit gates, and use $O(\alpha^{-2} \cdot n\log{n})$ ancillary qubits.

\paragraph{Large fields:}
For the case where the field size is large (say, $\abs{\F} \geq 10/\alpha$), there is a simpler worst case to average case reduction. Note that since we can efficiently verify a candidate solution using \cref{lem:qVerify-simple}, we can sample a line $\ell$ passing through the input vector $v$ and a uniformly random point $x \in \F^n$. It is easy to show that with probability at least $\Omega(\alpha)$, at least an $\Omega(\alpha)$-fraction of the points on $\ell$ are in the good set $X$, i.e., the set of points on which $\ALG^M$ outputs the correct answer with probability at least $\Omega(\alpha)$. Thus, with probability $\Omega(\alpha^3)$, we  sample two random points $a, b$ on $\ell$ that both belong to $X$, and then we can interpolate the value of $M \cdot v$ from $\ALG^M(a)$ and $\ALG^M(b)$. We note that this interpolation technique inherently requires large fields, while the proof presented above works over all finite fields.

\paragraph{Making $\ALG'$ a unitary:}
Finally, we note that it is straightforward to turn the algorithm $\ALG'$, which is produced by our reduction, into a unitary quantum algorithm. Yet, it is important to do this carefully so as to keep the overheads in gate complexity of the resulting unitary within the desired bound of $O(n^{3/2})$. We sketch how to do this below.

Suppose we have performed Step~\ref{item:learning-X*} and obtained a classical description of the set of vectors $B$ and the corresponding indices $k_1,\ldots,k_t$. We now need to perform Steps~5-9 in a unitary manner. In lieu of Step~ \ref{item:sample-from-X*}, we maintain three copies of the output state of the unitary $\mathsf{Q_{samp}}$ of \cref{cor:quantum-sampling} in three quantum registers, one each for the vectors $x_1,x_2$ and $x_3$. Step~\ref{item:self-correcting-X*} only involves arithmetic operations over $\F$, and we perform the required computation by a unitary circuit that computes the sparse shift vector in the form of a state $\ket{s}$, using $B$ and $k_1,\ldots,k_t$. In an ancillary register, we then compute a superposition of vectors corresponding to $v-x_1-x_2-x_3-s$, using the states representing the uniform samples $x_1,x_2$ and $x_3$.

By using amplitude amplification, we can construct a unitary implementation corresponding to $\ALG_{boost}^M$ with $O(\frac{1}{\sqrt{\alpha}}\log\frac{1}{\delta})$ uses of $\ALG_{\mathrm{verified}}$. Next, we apply this unitary to each of the four registers corresponding to $x_1,\ldots,x_4$, and also compute $M\cdot s$ using the circuit described in \cref{sec:quantumtools-verification}. Finally, we perform more arithmetic to compute a state corresponding to the vector $b$ of Step~\ref{item:computing-Mv} in an ancillary register. Applying $\mathsf{Q_{verify}}$ from \cref{lem:qVerify-simple} to the registers containing $v$ and $b$, we obtain a version of $\ALG'$ which attaches a flag indicating success to its output register. This can then be boosted using $O(1/\alpha)$ rounds of fixed-point amplitude amplification to obtain the desired unitary quantum algorithm.

\printbibliography

\appendix
\newpage
\section{Quantum singular value transformation techniques}
\label{app:qsvt}

In this appendix, we provide self-contained statements of techniques that make use of the quantum singular value transformations literature.

\subsection{Fixed-point amplitude amplification}
An issue that we frequently face while using standard amplitude amplification is the inability to \textit{a priori} estimate the number of iterations to perform without overshooting the target state. When the goal is to output a classical solution to a search problem, one can sidestep this issue by iterating over the number of rounds of amplification with a multiplicative step size (often called exponential search). But if the goal of amplification is to obtain a unitary subroutine that can later be composed with another unitary, this problem can be addressed by the technique of fixed-point amplitude amplification, which generalises the notion of fixed-point search \cite{Grover2005FixedPointSearch} and converges monotonically towards the marked state as the number of iterations is increased. We use this machinery primarily in the proof of \cref{lem:qVerify-simple}.


The precise version of fixed-point amplitude amplification that we need is expressed in the framework of quantum singular value transformations \cite[Theorem~27, arXiv version]{Gilyen2019}, as follows.

\begin{theorem}[Fixed-point amplitude amplification]
\label{thm:fpaa}
    Let $U$ be a unitary acting on $k+1$ qubits and $\Pi$ an orthogonal projector. If for an input state $\ket{\phi_{\mathrm{in}}}$ and $p>\delta>0$ we have $\Pi U\ket{\phi_{\mathrm{in}}}=p\ket{\phi_{\mathrm{tar}}}$, then for every $\epsilon>0$ there is a unitary $U_q$ such that $\norm{\ket{\phi_{\mathrm{tar}}}-U_q\ket{\phi_{\mathrm{in}}}}\leq \epsilon$. This $U_q$ requires $q=O(\frac{1}{\delta}\log\frac{1}{\epsilon})$ uses of $U$ and $U^{\dagger}$, one ancillary qubit, and $O(q)$ additional one- and two-qubit gates.
\end{theorem}

\subsection{Singular value threshold projections} To obtain our unitary implementation of the (noisy) indicator function in \cref{sec:indicator}, we need two conditions to be simultaneously satisfied: (1) the singular values that are larger than a threshold $t$ are boosted close to unity, \textit{and} (2) the singular values that are below the threshold are suppressed close to zero. In particular, fixed-point amplitude amplification only guarantees the former condition and is not the right tool for the task. We hence invoke a more sophisticated tool, namely, quantum singular value threshold projection. We quote the following result that we use~\cite[Theorem~31, arXiv version]{Gilyen2019}.

\begin{theorem}[Singular value threshold projections]
\label{thm:qsv-projections}
    Let $U$ be a unitary acting on $k+1$ qubits and $\Pi,\widetilde{\Pi}$ be orthogonal projectors. Suppose 
    \[
        \widetilde{\Pi}U\Pi = \sum_{j=1}^m \zeta_j\ketbra{w_j}{v_j}
    \]
    for $\zeta_j\in(0,1)$ and $\{w_j\}$ and $\{v_j\}$ two sets of orthonormal vectors. Then for any threshold $t\in(0,1)$ and $\epsilon,\delta>0$, there exists a unitary $U_q$ that makes $q=O\left(\frac{1}{\delta}\log\frac{1}{\epsilon}\right)$ queries to $U$ and $U^{\dagger}$, and uses $O(q)$ additional one- and two-qubit gates, such that 
    \begin{align}
        \widetilde{\Pi}U_q\Pi = \sum_{j=1}^m \zeta'_j\ketbra{w_j}{v_j},~~\text{ s.t. }\forall j\in [m],~~\begin{cases}
            \zeta'_j \geq 1-\epsilon & \zeta_j \in [t+\delta/2,1]\\
        |\zeta'_j| \leq 1 \hphantom{1-\epsilon}\text{ if }& \zeta_j \in (t-\delta/2, t+\delta/2)\\
        \zeta'_j \leq \epsilon & \zeta_j \in [0,t-\delta/2]
        \end{cases}
    \end{align}
    %
\end{theorem}

\paragraph{Computing the circuits for $U_q$.} The two theorems stated above are in fact constructive and provide explicit circuits for the unitaries $U_q$. In the interest of space, we simply note here that there is a classical runtime overhead of $O(q^3\poly\log(q/\epsilon))$, and refer the interested reader to \cite{Gilyen2019,Haah2019} for details. For our applications, this overhead is of order $\widetilde{O}(n^{3/2})$.

\end{document}